\documentclass[
  aps,pra,reprint,
  amsmath,amssymb,
  floatfix,superscriptaddress,longbibliography
]{revtex4-2}

\usepackage{comment}
\usepackage[T1]{fontenc}
\usepackage[utf8]{inputenc}
\usepackage{amsthm,mathtools}
\usepackage{bm}

\usepackage{physics}
\usepackage{graphicx}
\usepackage{microtype}
\usepackage{xcolor}
\usepackage{enumitem}
\usepackage[colorlinks=true,linkcolor=blue,citecolor=teal,urlcolor=blue]{hyperref}
\hypersetup{
  pdftitle={Flux control of measurement back-action and Leggett--Garg correlations in chiral quantum walks},
  pdfauthor={Paolo Luppi}
}
\usepackage{mathrsfs}
\theoremstyle{plain}
\newtheorem{theorem}{Theorem}
\newtheorem{proposition}[theorem]{Proposition}

\newtheorem{corollary}[theorem]{Corollary}
\theoremstyle{definition}

\theoremstyle{remark}

\let\Tr\relax
\DeclareMathOperator{\Tr}{Tr}

\begin{document}

\title{Flux control of measurement back-action and Leggett–Garg correlations in chiral quantum walks}

\author{Paolo Luppi}
\email{paolo.luppi@unimi.it}
\affiliation{Dipartimento di Fisica “Aldo Pontremoli”, Università degli Studi di Milano, via Celoria 16, I-20133 Milan, Italy}
\affiliation{INFN, Sezione di Milano, via Celoria 16, I-20133 Milan, Italy}

\begin{abstract}
Gauge-invariant fluxes control interference in chiral continuous-time quantum walks. We investigate how they affect sequential measurements at a single vertex, using a dichotomic observable that distinguishes return to that vertex from occupation of its complement. For a walker initially localized at the measured vertex, the complete two-time statistics, including measurement back-action and Leggett--Garg correlators, are determined exactly by the return amplitude, connecting temporal correlations to the local spectral measure and gauge-invariant closed-walk interference. At short times, the leading disturbance is independent of the Peierls phases, whereas flux sensitivity enters at higher orders through interference among closed walks. We further identify a graph-independent sufficient mechanism for saturating the L\"uders bound: flux can reduce the rooted dynamics to a balanced two-dimensional Krylov subspace with equal spectral weights, yielding a constructive flux-engineering criterion for attaining the L\"uders bound of $3/2$ at finite times. The mechanism is realized exactly on a two-flux diamond graph, where destructive interference renders additional rooted modes dark and the local back-action depends on relative combinations of the two independent fluxes. For flux-threaded cycles, an exact winding-number expansion reveals a parity-dependent onset: the leading flux contrast occurs generically at order $t^N$ for even cycles and $t^{2N}$ for odd cycles. Across the cycles examined, flux can either enhance the maximal Leggett--Garg violation or shift strong violations to earlier measurement times, with half flux driving the four-site cycle to the L\"uders bound. These results establish gauge-invariant flux as a resource for engineering local measurement back-action and temporal quantum correlations.
\end{abstract}

\maketitle

\section{Introduction}
\label{sec:introduction}

Continuous-time quantum walks (CTQWs) provide a framework for coherent
transport on discrete structures, with applications to quantum algorithms,
quantum simulation, transport phenomena, and graph-based information
processing \cite{Portugal,mulken2011,Apers2022,QW2particles}. Their
Hamiltonian encodes the graph connectivity and weights, so that
interference between propagation paths reflects both local and global
structure. Complex hopping amplitudes define chiral CTQWs and realize
synthetic gauge fields on graphs
\cite{Peierls1933,Luttinger1951,Harper1955,Lu2016,Cedzich2019,
Frigerio2021Generalized,Aidelsburger2018,Razzoli2020}. They can break
time-reversal symmetry, generate directional currents, and enable
interference-based control of transport, state transfer, routing,
recurrence, graph inference, and quantum batteries
\cite{Zimboras2013QuantumTransportEnhancement,Yu2023ControlledTransport,
Acuaviva2025StateTransfer,Bottarelli2023,Annoni2024,
Frigerio2023SwiftChiralQuantumWalks,Forghieri2025,
Cavazzoni2025PerfectRouting,Ragazzi2025,Benedetti2026NoisyRouting,
Cavazzoni2026}. Individual edge phases are gauge dependent, however, and
observable effects can depend only on the gauge-invariant fluxes
accumulated along closed cycles.

Single-time quantities, such as transition and return probabilities or
transport efficiencies, are standard diagnostics of quantum-walk
dynamics. Sequential measurements provide complementary information by
comparing statistics obtained under different measurement schedules.
The resulting multitime distributions may violate the Kolmogorov
consistency conditions satisfied by classical stochastic processes
\cite{KoflerBrukner2013,Clemente2015,Halliwell2017,Smirne2019}. Such
violations quantify measurement disturbance and are closely related to
no-signaling-in-time conditions and Leggett--Garg tests of macrorealism
under the usual noninvasiveness assumptions
\cite{LG1985,Lambert2014,VitaglianoBudroni2023}. The two diagnostics are
nevertheless logically distinct in general: exact pairwise no-signaling
in time can coexist with a Leggett--Garg violation for suitable initial
states \cite{Kim2026ExactNoSignaling}. Both have been investigated
experimentally in quantum walks
\cite{Smirne2020ExperimentalControl,Robens2015} and connected to
randomness generation and quantum-device benchmarks
\cite{Nath2024LGIRandomness,Rybotycki2026LGIQuantumComputers}.

Previous work on sequential measurements in CTQWs generated by real graph
Laplacians showed that single-time and multitime signatures can respond
differently to graph structure \cite{Luppi2026}. It remains unclear,
however, how gauge-invariant fluxes enter local sequential statistics,
at which time order their effects first appear, and how they modify
temporal correlations.

Here we consider a walker initially localized at a vertex $\nu$ and the
dichotomic return observable
$Q_\nu=2\ket{\nu}\bra{\nu}-\mathbb{I}$, which distinguishes return to
$\nu$ from occupation of its orthogonal complement. The associated
L\"uders instrument requires access to only one vertex and preserves
coherences within the complementary subspace, thereby probing return-path
interference without a fully resolved position measurement. Related
return projectors have been used in studies of quantum hitting, first
detection, recurrence, and transport optimization
\cite{Varbanov2008HittingTime,Friedman2017QuantumWalks,Roy2025,
Finocchiaro2025OptimalQuantumTransport}, including flux-controlled
recurrence under stroboscopic monitoring
\cite{Thiel2020DarkStates,Yin2025Resonances}. Here, instead, a single
intermediate measurement is used to quantify its effect on
subsequent return statistics and temporal correlations.

We show that the complete two-time statistics are determined exactly by
the return amplitude rooted at the measured vertex, connecting
measurement back-action and Leggett--Garg correlators to the rooted
spectral measure and to gauge-invariant closed-walk interference. For the localized preparation considered here, the measurement-induced change $K_\nu$ in the return probability bounds the Leggett--Garg functional, $L_3\leq 1+2K_\nu$, at the corresponding measurement times, so that any Leggett--Garg violation necessarily entails nonzero back-action. At short times, the leading
disturbance is independent of the Peierls phases, while flux sensitivity
begins at the lowest even order for which the complete sum of rooted
closed-walk contributions does not cancel.

We further identify a graph-independent sufficient mechanism for
saturating the L\"uders bound $L_3=3/2$ \cite{Fritz2010,Budroni2013,BudroniEmary2014}. If the dynamics accessible from
the measured vertex reduces to a two-dimensional rooted Krylov subspace,
the maximal violation is fixed by the two rooted spectral weights and
reaches the bound when they are equal. Gauge flux can engineer this balanced effective two-level dynamics by modifying both the rooted spectral support and its degeneracy structure. The two-flux diamond graph provides an exact multicycle
realization: its local statistics depend on relative combinations of the
independent fluxes, and an appropriate configuration saturates the bound
at an outer vertex.

For flux-threaded cycles, an exact winding-number expansion yields a
parity-dependent onset of the generic flux contrast, scaling as $t^N$
for even cycles and as $t^{2N}$ for odd cycles. Half flux drives $C_4$ to $L_3=3/2$, whereas in the larger cycles and finite observation
windows examined numerically its main effect is to shift strong
violations to earlier measurement times. Flux-threaded cycles provide analytically tractable benchmarks for the exact return-amplitude framework and the rooted-Krylov criterion, both of which apply to arbitrary finite graphs.

Section~\ref{sec:model} introduces the chiral walk and the measurement
protocol. Sections~\ref{sec:exact-statistics} and
\ref{sec:flux-onset} derive the exact statistics and their short-time
flux dependence. Section~\ref{sec:LGI} develops the Leggett--Garg
analysis and the rooted-Krylov criterion, while Sec.~\ref{sec:cycle}
treats flux-threaded cycles. Section~\ref{sec:conclusions} summarizes
the results and outlines extensions to open-system dynamics and flux
metrology.

\section{Chiral continuous-time quantum walks and sequential measurements}
\label{sec:model}

\subsection{Chiral continuous-time quantum walks}
\label{subsec:chiral_walks}

A continuous-time quantum walk (CTQW) describes the unitary evolution of a
quantum particle on a graph $G=(V,E)$ \cite{Portugal,mulken2011}. The Hilbert space is spanned by the
localized basis $\{\ket{x}\}_{x\in V}$, where $\ket{x}$ represents occupation
of vertex $x$. Throughout the paper we set $\hbar=1$.

In the Laplacian convention adopted here, the dynamics is generated by
\begin{equation}
H=D-A,
\label{eq:achiral_hamiltonian}
\end{equation}
where $D$ is the degree matrix and $A$ is the adjacency matrix. Since $H$ is
real and symmetric, the transition probabilities
\begin{equation}
p_{j\rightarrow k}(t)
=
\left|
\bra{k}e^{-iHt}\ket{j}
\right|^2
\end{equation}
satisfy
\begin{align}
p_{j\rightarrow k}(t)
&=
p_{j\rightarrow k}(-t),
\label{eq:time_symmetry}
\\
p_{j\rightarrow k}(t)
&=
p_{k\rightarrow j}(t),
\label{eq:reciprocity}
\end{align}
expressing time-reversal invariance and reciprocity, respectively.

Chirality is introduced by assigning an oriented Peierls phase
$\theta_{jk}=-\theta_{kj}$ to each edge \cite{Peierls1933,Luttinger1951,Harper1955}. The resulting complex adjacency
matrix is
\begin{equation}
(A_\chi)_{jk}
=
\begin{cases}
e^{i\theta_{jk}}, & \{j,k\}\in E,\\
0, & \text{otherwise},
\end{cases}
\label{eq:chiral_adjacency}
\end{equation}
and the chiral walk is generated by
\begin{equation}
H_\chi=D-A_\chi.
\label{eq:chiral_hamiltonian}
\end{equation}
The antisymmetry of the phases ensures that $H_\chi$ is Hermitian.
Complex hopping amplitudes modify the interference between propagation
paths and can break both time-reversal symmetry and reciprocity, thereby
producing directional transport and nonvanishing probability currents
\cite{Peierls1933,Luttinger1951,Harper1955,Lu2016,Cedzich2019}.

The individual edge phases are gauge dependent. Consider the diagonal
unitary transformation
\begin{equation}
\Lambda_{\boldsymbol{\beta}}
=
\sum_{j\in V}
e^{i\beta_j}\ket{j}\!\bra{j}.
\end{equation}
The Hamiltonian transforms as
\begin{equation}
H_\chi
\longmapsto
H_\chi'
=
\Lambda_{\boldsymbol{\beta}}
H_\chi
\Lambda_{\boldsymbol{\beta}}^\dagger,
\label{eq:gauge_transformation}
\end{equation}
which corresponds to the phase transformation
\begin{equation}
\theta_{jk}
\longmapsto
\theta_{jk}+\beta_j-\beta_k.
\label{eq:edge_phase_transformation}
\end{equation}
The transition amplitudes acquire only vertex-dependent phase factors,
and hence the transition probabilities are gauge invariant:
$\left|\bra{k}e^{-iH_\chi't}\ket{j}\right|^2
=
\left|\bra{k}e^{-iH_\chi t}\ket{j}\right|^2$.

Physical observables can consequently depend only on gauge-invariant
combinations of edge phases. For an oriented cycle
$C=(j_1,j_2,\ldots,j_m,j_1)$ of length $m$, the corresponding flux is
\begin{equation}
\Phi_C
=
\sum_{r=1}^{m}
\theta_{j_{r+1}j_r}
\pmod{2\pi},
\qquad
j_{m+1}\equiv j_1.
\label{eq:gauge_flux}
\end{equation}
For a connected graph, the independent physical phase parameters may be
chosen as the fluxes associated with a basis of independent cycles \cite{Lu2016,Cedzich2019}. We
denote their collection by $\boldsymbol{\Phi}$ and write the Hamiltonian
and propagator as
\begin{equation}
H_\chi(\boldsymbol{\Phi}),
\qquad
U_t(\boldsymbol{\Phi})
=
e^{-iH_\chi(\boldsymbol{\Phi})t}.
\end{equation}

Time reversal complex conjugates the hopping amplitudes and maps
$\boldsymbol{\Phi}\mapsto-\boldsymbol{\Phi}$. In a graph with a single
independent cycle, the fluxes $\Phi=0$ and $\Phi=\pi$ are invariant under
this transformation modulo $2\pi$, whereas generic values of $\Phi$ break
time-reversal symmetry. We will therefore distinguish between dependence
on a gauge-invariant flux and genuine time-reversal-symmetry breaking.

\subsection{Single-node sequential measurements}
\label{subsec:measurement_protocol}

We consider a walker initially localized at a vertex $\nu$, $\rho_0=\ket{\nu}\!\bra{\nu}$.
To probe the dynamics locally, we use the dichotomic return observable
\begin{equation}
Q_\nu
=
2\Pi_\nu-\mathbb{I},
\qquad
\Pi_\nu=\ket{\nu}\!\bra{\nu}.
\label{eq:return_observable}
\end{equation}
Its outcomes distinguish return to the initial vertex,
$q=+1$, from occupation of the complementary subspace,
$q=-1$. The corresponding projectors are $\Pi_{+}=\Pi_\nu$ and $\Pi_{-}=\mathbb{I}-\Pi_\nu$.
The measurement protocol also differs essentially from that adopted in
Ref.~\cite{Luppi2026}. There, the intermediate measurement resolved the
complete position observable and the associated nonselective channel
removed all coherences in the site basis. Here we instead use the
dichotomic Lüders instrument
$\{\Pi_\nu,\mathbb{I}-\Pi_\nu\}$, which distinguishes return to a single
reference vertex from its complement while preserving coherences within
the complementary subspace. The dichotomic observable is a coarse graining of position, but its L\"uders instrument is not operationally equivalent to a fully resolved position measurement followed by outcome coarse graining. It requires access to only one vertex,
reduces the two-time statistics exactly to the rooted return amplitude, and
directly provides the dichotomic correlators entering the Leggett--Garg
inequalities.
The rooted return amplitude and probability are defined as
\begin{align}
A_\nu(t;\boldsymbol{\Phi})
&=
\bra{\nu}
U_t(\boldsymbol{\Phi})
\ket{\nu},
\label{eq:return_amplitude}
\\
p_\nu(t;\boldsymbol{\Phi})
&=
\left|
A_\nu(t;\boldsymbol{\Phi})
\right|^2.
\label{eq:return_probability}
\end{align}
These quantities are gauge invariant and are determined by the spectral
measure of $H_\chi$ rooted at $\nu$.

For sequential measurements at times
$t_1<t_2<\cdots<t_n$, we introduce the unitary channel
\begin{equation}
\mathcal{U}_t[\rho]
=
U_t\rho U_t^\dagger
\end{equation}
and the measurement operations
\begin{equation}
\mathcal{M}_{q}[\rho]
=
\Pi_q\rho\Pi_q,
\qquad
q\in\{+1,-1\}.
\end{equation}
The joint probability of obtaining the sequence of outcomes
$q_1,\ldots,q_n$ is then
\begin{equation}
\begin{split}
P(q_n,t_n;\ldots;q_1,t_1)
=
\Tr\Big[
\mathcal{M}_{q_n}
\mathcal{U}_{t_n-t_{n-1}}
\cdots
\mathcal{M}_{q_1}
\mathcal{U}_{t_1}[\rho_0]
\Big].
\end{split}
\label{eq:multitime_probability}
\end{equation}

To quantify measurement back-action, we compare the return probability at
time $t$ with and without a nonselective measurement at an intermediate
time $s<t$. The nonselective Lüders channel is
\begin{equation}
\mathcal{D}_\nu[\rho]
=
\Pi_{+}\rho\Pi_{+}
+
\Pi_{-}\rho\Pi_{-},
\label{eq:luders_channel}
\end{equation}
and the measured return probability is
\begin{equation}
p_\nu^{(s)}(t;\boldsymbol{\Phi})
=
\Tr\!\left[
\Pi_\nu
\mathcal{U}_{t-s}
\mathcal{D}_\nu
\mathcal{U}_{s}[\rho_0]
\right].
\label{eq:measured_return_probability}
\end{equation}
The corresponding Kolmogorov inconsistency is
\begin{equation}
K_{\nu,\boldsymbol{\Phi}}(s,t)
=
\left|
p_\nu^{(s)}(t;\boldsymbol{\Phi})
-
p_\nu(t;\boldsymbol{\Phi})
\right|.
\label{eq:kolmogorov_return}
\end{equation}
For the dichotomic return measurement, this is the total-variation distance
between the final outcome distributions obtained with and without the
intermediate measurement \cite{Clemente2015,Halliwell2017,Smirne2020ExperimentalControl,Luppi2026}.

Beyond the disturbance of the final-time marginal, sequential measurements also define temporal correlations between outcomes observed at different times. These correlations enter the Leggett--Garg inequalities \cite{LG1985,Lambert2014}, often regarded as temporal analogues of Bell inequalities, although their underlying assumptions differ substantially. They provide tests of macrorealistic descriptions of the dynamics. The two-time probabilities determine the temporal correlators \begin{equation} C_{ij} = \sum_{q_i,q_j=\pm1} q_i q_j\, P(q_j,t_j;q_i,t_i), \label{eq:two_time_correlator} \end{equation} from which we construct the three-time Leggett--Garg functional \begin{equation} L_3 = C_{12}+C_{23}-C_{13}. \label{eq:L3_definition} \end{equation} Under the standard assumptions of macrorealism and noninvasive measurability, the Leggett--Garg inequality reads \begin{equation} L_3\leq 1. \end{equation} A violation therefore rules out the simultaneous validity of these assumptions for the observed statistics \cite{LG1985,Lambert2014,VitaglianoBudroni2023,Robens2015}.
 In the following, we use $K_{\nu,\boldsymbol{\Phi}}$ and $L_3$
as complementary probes of the way gauge-invariant fluxes modify
measurement disturbance and temporal correlations.  
\section{Exact sequential return statistics}
\label{sec:exact-statistics}

The combination of the localized initial state, the rank-one return
projector, and time-homogeneous unitary evolution implies that the complete
two-time joint distribution can be expressed solely in terms of the return
amplitude at the measured vertex. The following result holds for any
finite graph, arbitrary hopping phases, and any time-independent Hamiltonian.

\begin{proposition}[Exact two-time return statistics]
\label{prop:two_time_probabilities}
Let the walker be initially localized at $\nu$, and let measurements of
$Q_\nu$ be performed at times $t_1<t_2$. Defining
$\tau=t_2-t_1$ and
\begin{equation}
P_{q_1q_2}(t_1,t_2)
:=
P(q_2,t_2;q_1,t_1),
\qquad
q_1,q_2\in\{+,-\},
\label{eq:joint_probability_notation}
\end{equation}
where the first and second indices refer, respectively, to the outcomes at
$t_1$ and $t_2$, the joint probabilities are
\begin{align}
P_{++}(t_1,t_2)
&=
p_\nu(t_1)\,p_\nu(\tau),
\label{eq:Ppp}
\\
P_{+-}(t_1,t_2)
&=
p_\nu(t_1)
\left[
1-p_\nu(\tau)
\right],
\label{eq:Ppm}
\\
P_{-+}(t_1,t_2)
&=
\left|
A_\nu(t_2)
-
A_\nu(\tau)A_\nu(t_1)
\right|^2,
\label{eq:Pmp}
\\
P_{--}(t_1,t_2)
&=
1-p_\nu(t_1)
-
\left|
A_\nu(t_2)
-
A_\nu(\tau)A_\nu(t_1)
\right|^2.
\label{eq:Pmm}
\end{align}
Here and below, the dependence on $\boldsymbol{\Phi}$ is left implicit
whenever no ambiguity arises.
\end{proposition}

A derivation of Eqs.~\eqref{eq:Ppp}--\eqref{eq:Pmm} is given in
Appendix~\ref{app:exact-statistics}. The nontrivial term
$P_{-+}$ describes return to $\nu$ after projection onto the complementary
subspace at the first measurement. It is an interference quantity
involving the uninterrupted return amplitude at $t_2$ and the product of
the return amplitudes over the two successive time intervals.

For a nonselective intermediate measurement at time $s<t$, summing over
the two possible outcomes gives
\begin{equation}
\begin{split}
p_\nu^{(s)}(t;\boldsymbol{\Phi})
={}&
p_\nu(s;\boldsymbol{\Phi})
p_\nu(t-s;\boldsymbol{\Phi})
\\
&+
\left|
A_\nu(t;\boldsymbol{\Phi})
-
A_\nu(t-s;\boldsymbol{\Phi})
A_\nu(s;\boldsymbol{\Phi})
\right|^2.
\end{split}
\label{eq:exact_measured_return}
\end{equation}
The two terms correspond, respectively, to return after the intermediate
outcomes $+1$ and $-1$. The Kolmogorov inconsistency introduced in
Eq.~\eqref{eq:kolmogorov_return} can equivalently be written as
\begin{equation}
\begin{split}
K_{\nu,\boldsymbol{\Phi}}(s,t)
=
2\Big|
&p_\nu(s;\boldsymbol{\Phi})
p_\nu(t-s;\boldsymbol{\Phi})
\\
&-
\operatorname{Re}\!\left[
A_\nu(t;\boldsymbol{\Phi})^*
A_\nu(t-s;\boldsymbol{\Phi})
A_\nu(s;\boldsymbol{\Phi})
\right]
\Big|.
\end{split}
\label{eq:exact_K}
\end{equation}
Thus, the measurement back-action is determined by the mismatch between
the product of two return probabilities and the corresponding coherent
concatenation of return amplitudes. Equation~\eqref{eq:exact_K} provides
the starting point for the short-time analysis in
Sec.~\ref{sec:flux-onset}.

For any pair $t_i<t_j$, evaluated in an independent two-time run,
$C_{ij}=1-2P_{+-}(t_i,t_j)-2P_{-+}(t_i,t_j)$. Using
Eqs.~\eqref{eq:Ppm} and \eqref{eq:Pmp}, one obtains
\begin{equation}
\begin{split}
C_{ij}
={}&
1
-
2p_\nu(t_i)
\left[
1-p_\nu(t_j-t_i)
\right]
\\
&-
2\left|
A_\nu(t_j)
-
A_\nu(t_j-t_i)A_\nu(t_i)
\right|^2.
\end{split}
\label{eq:exact_correlator}
\end{equation}
Although the Hamiltonian is time independent, the correlator generally
depends separately on $t_i$ and $t_j$, because of the localized initial
condition and the state update induced by the first measurement.

\paragraph{Global flux-reversal invariance.}

For the class of chiral CTQWs considered here, the simultaneous reversal
of all gauge-invariant fluxes gives
$H_\chi(-\boldsymbol{\Phi})=H_\chi(\boldsymbol{\Phi})^*$ in the vertex
basis. Consequently,
\begin{equation}
\begin{split}
A_\nu(t;-\boldsymbol{\Phi})
&=
\left[
\bra{\nu}
e^{iH_\chi(\boldsymbol{\Phi})t}
\ket{\nu}
\right]^*
=
\left[
A_\nu(-t;\boldsymbol{\Phi})
\right]^*
\\
&=
A_\nu(t;\boldsymbol{\Phi}),
\end{split}
\label{eq:A_flux_even}
\end{equation}
where Hermiticity was used in the last equality.

\begin{corollary}[Global flux-reversal invariance]
\label{cor:flux_reversal}
All return-based quantities considered in this work are invariant under
the simultaneous reversal of all independent fluxes:
\begin{align}
p_\nu(t;-\boldsymbol{\Phi})
&=
p_\nu(t;\boldsymbol{\Phi}),
&
K_{\nu,-\boldsymbol{\Phi}}(s,t)
&=
K_{\nu,\boldsymbol{\Phi}}(s,t),
\nonumber
\\
C_{ij}(-\boldsymbol{\Phi})
&=
C_{ij}(\boldsymbol{\Phi}),
&
L_3(-\boldsymbol{\Phi})
&=
L_3(\boldsymbol{\Phi}).
\label{eq:L3_flux_even}
\end{align}
\end{corollary}

The single-node protocol can therefore detect flux-induced modifications
of the dynamics, but it cannot distinguish a flux configuration
$\boldsymbol{\Phi}$ from its global reversal $-\boldsymbol{\Phi}$. For a
graph with one independent flux, the return statistics are even and
$2\pi$-periodic functions of $\Phi$. The points $\Phi=0$ and $\Phi=\pi$
are invariant under time reversal modulo $2\pi$, whereas generic values
break time-reversal symmetry while producing the same return statistics
as the oppositely oriented flux $-\Phi$.

In a multicycle graph, global flux-reversal invariance does not imply
invariance under the reversal of only one independent flux. Return
statistics may depend on relative combinations of the fluxes, and a
partial reversal can therefore modify the observed dynamics. This point
will be illustrated explicitly for the diamond graph in
Sec.~\ref{subsec:diamond}.

To quantify the modification induced by the gauge-invariant fluxes
relative to the zero-flux dynamics, we define
\begin{equation}
\Delta_{\boldsymbol{\Phi}}K_\nu(s,t)
=
K_{\nu,\boldsymbol{\Phi}}(s,t)
-
K_{\nu,\boldsymbol{0}}(s,t),
\label{eq:flux_K_contrast}
\end{equation}
where $\boldsymbol{0}$ denotes the configuration in which all independent
cycle fluxes vanish.

\section{Gauge-invariant onset of flux sensitivity}
\label{sec:flux-onset}

We now use the exact expression in Eq.~\eqref{eq:exact_K} to determine
how gauge-invariant fluxes enter the short-time sequential statistics.
We introduce the final evolution time $t$ and place the intermediate
measurement at a fixed fraction of it, $s=\alpha t$ with $0<\alpha<1$,
so that the second evolution interval is $t-s=(1-\alpha)t$.

For a fixed measured vertex $\nu$, we define the rooted spectral moments
\begin{equation}
\mu_n^{(\nu)}(\boldsymbol{\Phi})
=
\bra{\nu}
H_\chi(\boldsymbol{\Phi})^n
\ket{\nu}.
\label{eq:rooted_moments}
\end{equation}
Since $H_\chi$ is Hermitian, all
$\mu_n^{(\nu)}$ are real. We also introduce the centered Hamiltonian and the corresponding centered
moments,
\begin{equation}
\widetilde H_\nu
=
H_\chi-\mu_1^{(\nu)}\mathbb{I},
\qquad
\widetilde\mu_n^{(\nu)}
=
\bra{\nu}\widetilde H_\nu^{\,n}\ket{\nu}.
\label{eq:centered_definitions}
\end{equation}
In particular,
\begin{equation}
\gamma_\nu
:=
\widetilde\mu_2^{(\nu)}=
\mu_2^{(\nu)}
-
\left(\mu_1^{(\nu)}\right)^2=
\sum_{x\neq\nu}
\left|
(H_\chi)_{x\nu}
\right|^2.
\label{eq:local_coupling_strength}
\end{equation}
For unit-modulus hopping on an unweighted graph, $\gamma_\nu$ reduces to
the degree of the measured vertex.

\subsection{Even short-time structure}
\label{subsec:even_expansion}

Hermiticity implies
$A_\nu(-t;\boldsymbol{\Phi})
=
A_\nu(t;\boldsymbol{\Phi})^*$.
Consequently, both the measured and unmeasured return probabilities
are invariant under the simultaneous reversal of all time intervals.

\begin{proposition}[Even short-time structure]
\label{prop:even_short_time}
Let $\gamma_\nu>0$ and $0<\alpha<1$. In a neighborhood of $t=0$,
the return-based Kolmogorov inconsistency admits the expansion
\begin{equation}
K_{\nu,\boldsymbol{\Phi}}(\alpha t,t)
=
\sum_{m=1}^{\infty}
k_{2m}^{(\nu)}
(\alpha,\boldsymbol{\Phi})t^{2m}.
\label{eq:even_K_expansion}
\end{equation}
In particular, the short-time expansion contains only even powers of $t$.
\end{proposition}

The proof is given in Appendix~\ref{app:short-time}. If
$\gamma_\nu=0$, the state $\ket{\nu}$ is an eigenstate of the Hamiltonian
and $K_{\nu,\boldsymbol{\Phi}}(s,t)=0$
for all $s$ and $t$.

\paragraph{Leading flux-dependent correction.}
\label{subsec:quartic_expansion}

We define
\begin{equation}
\zeta_\nu
=
\frac{\widetilde\mu_4^{(\nu)}}{12}
+
\frac{\gamma_\nu^2}{4}.
\label{eq:chi_definition}
\end{equation}
The expansion of the exact expression
Eq.~\eqref{eq:exact_K} is

\begin{equation}
\begin{split}
K_{\nu,\boldsymbol{\Phi}}(\alpha t,t)
=& 2\gamma_\nu\alpha(1-\alpha)t^2
+
\Bigg\{
\zeta_\nu
\left[
\alpha^4+(1-\alpha)^4-1
\right]\\
&+
2\gamma_\nu^2
\alpha^2(1-\alpha)^2
\Bigg\}t^4
+
O(t^6).
\end{split}
\label{eq:K_quartic_expansion}
\end{equation}

The quadratic coefficient depends only on the moduli of the couplings
connected to $\nu$ and is therefore insensitive to the Peierls phases.
It coincides with the universal short-time law obtained for
time-reversal-invariant CTQWs. For real unit-weight Laplacians,
$\gamma_\nu=d_\nu$, so the leading quadratic term agrees with the full
position-measurement result of Ref.~\cite{Luppi2026}. Beyond this order,
however, the return protocol differs: its common-time expansion contains
only even powers, and its flux dependence is organized by
gauge-invariant rooted closed walks.

For fixed hopping strengths and on-site energies, $\gamma_\nu$ is
independent of the fluxes. The contrast with respect to the zero-flux
dynamics is consequently
\begin{align}
\Delta_{\boldsymbol{\Phi}}
K_\nu(\alpha t,t)
={}&
K_{\nu,\boldsymbol{\Phi}}(\alpha t,t)
-
K_{\nu,\boldsymbol{0}}(\alpha t,t)
\nonumber\\
={}&
\frac{
\Delta_{\boldsymbol{\Phi}}
\widetilde\mu_4^{(\nu)}
}{12}
\left[
\alpha^4+(1-\alpha)^4-1
\right]t^4
+
O(t^6),
\label{eq:flux_contrast_quartic}
\end{align}
where $\Delta_{\boldsymbol{\Phi}}
\widetilde\mu_4^{(\nu)}
=
\widetilde\mu_4^{(\nu)}(\boldsymbol{\Phi})
-
\widetilde\mu_4^{(\nu)}(\boldsymbol{0})$.
Thus, the fourth order is the earliest order at which flux sensitivity
can occur in the return-based Kolmogorov inconsistency. The quartic
coefficient need not be flux dependent, however: symmetries and
interference cancellations may postpone the first nonzero contribution to
sixth or higher order.

\paragraph{Rooted closed-walk interpretation.}
\label{subsec:closed_walks}

The rooted moments can be expanded as
\begin{equation}
\mu_n^{(\nu)}
=
\sum_{x_1,\ldots,x_{n-1}}
(H_\chi)_{\nu x_{n-1}}
(H_\chi)_{x_{n-1}x_{n-2}}
\cdots
(H_\chi)_{x_1\nu}.
\label{eq:moment_closed_walks}
\end{equation}
Each term represents a closed walk of length $n$ starting and ending at
$\nu$. Immediate backtracking carries no net phase, whereas walks winding
around graph cycles can acquire phases determined by the enclosed
gauge-invariant fluxes.

The detailed gauge transformation of these products is discussed in
Appendix~\ref{app:closed-walks}. The relevant observable coefficient emerges only after all closed-walk contributions at a given order are summed and combined into the corresponding centered moments. The shortest accessible cycle therefore
sets only a possible geometric scale for flux sensitivity; parity,
symmetries, and destructive interference can delay its observable onset.

Accordingly, the first flux-dependent term in
Eq.~\eqref{eq:even_K_expansion} is determined by the lowest even order at
which the corresponding gauge-invariant sum of rooted closed walks does
not cancel. In Sec.~\ref{sec:cycle}, this mechanism is evaluated
explicitly for flux-threaded cycles.
\subsection{A two-flux multicycle example}
\label{subsec:diamond}

The rooted closed-walk formulation shows that the onset of flux
sensitivity is not determined solely by the shortest cycle length. To
illustrate this point beyond vertex-transitive single-cycle systems, we
consider the diamond graph formed by two triangles sharing one edge,
\begin{equation*}
V=\{0,1,2,3\},
\qquad
E=
\bigl\{
\{0,1\},\{0,2\},\{1,2\},
\{0,3\},\{1,3\}
\bigr\}.
\label{eq:diamond_graph}
\end{equation*}
The graph is not vertex transitive (Fig.~\ref{fig:diamond_fluxes}): vertices $0$ and $1$ have degree
three, whereas vertices $2$ and $3$ have degree two. Its cycle-space
dimension is $|E|-|V|+1=2$. We choose the oriented triangles
$C_1=(0,1,2,0)$ and $C_2=(0,1,3,0)$ as a cycle basis, with corresponding fluxes $\Phi_1$ and $\Phi_2$. The
remaining simple cycle, $C_{\mathrm{out}}=(0,2,1,3,0)$, carries the
dependent flux $\Phi_{\mathrm{out}}=\Phi_2-\Phi_1$, up to the chosen
orientation. A convenient gauge gives the Laplacian Hamiltonian
\begin{equation}
H_{\diamond}(\Phi_1,\Phi_2)
=
\begin{pmatrix}
3 & -1 & -1 & -1\\
-1 & 3 & -e^{-i\Phi_1} & -e^{-i\Phi_2}\\
-1 & -e^{i\Phi_1} & 2 & 0\\
-1 & -e^{i\Phi_2} & 0 & 2
\end{pmatrix}.
\label{eq:diamond_hamiltonian}
\end{equation}
\begin{figure}[t]
    \centering
    \includegraphics[ width=\columnwidth ]{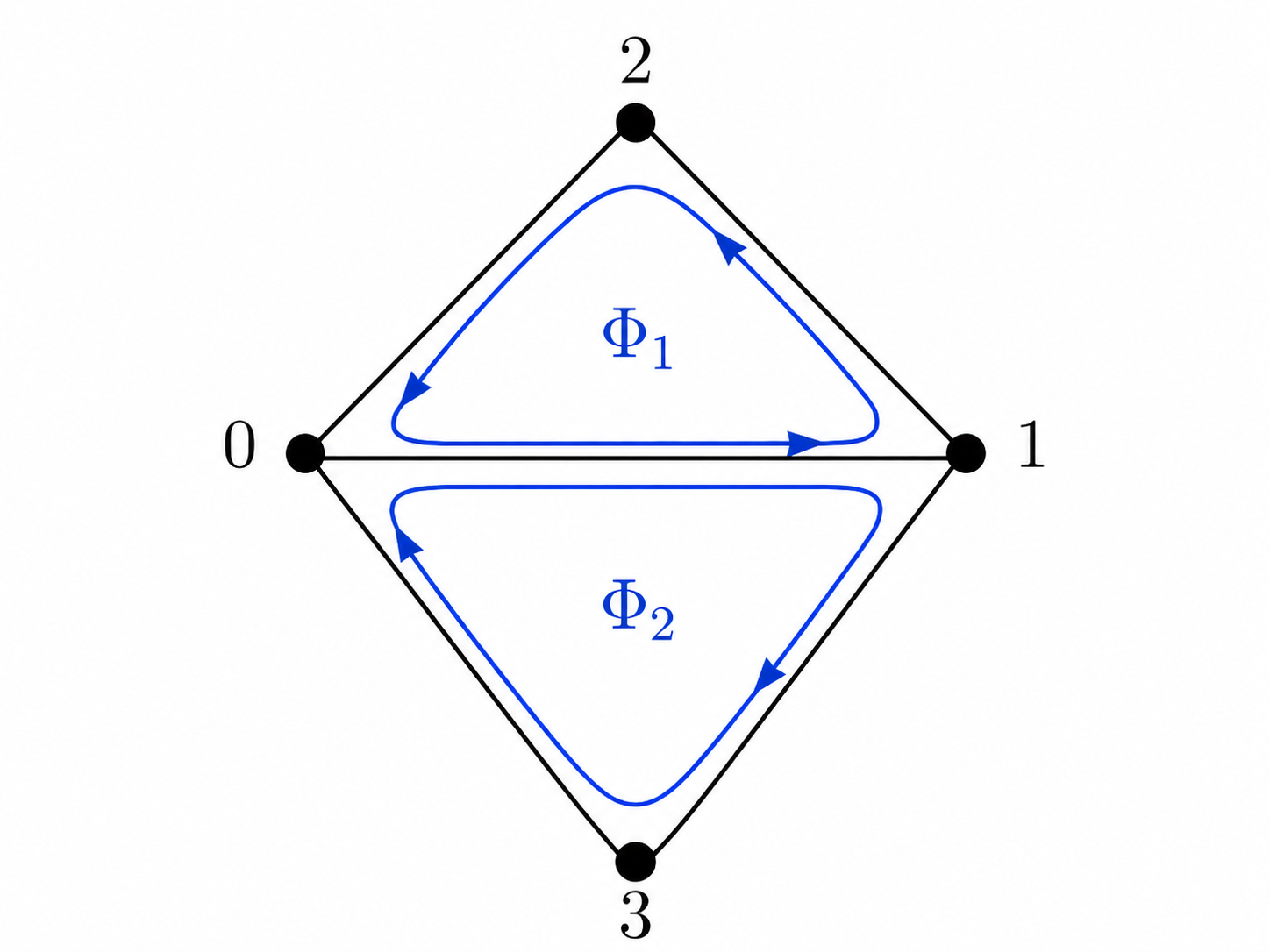}
    \caption{Diamond graph with the oriented fluxes $\Phi_1$ and
    $\Phi_2$ associated with the cycles
    $C_1=(0,1,2,0)$ and $C_2=(0,1,3,0)$, respectively.}
    \label{fig:diamond_fluxes}
\end{figure}
The contrast with an isolated odd cycle is already visible at fourth
order. For triangle $C_3$, the centered fourth
moment is independent of the flux,
$\widetilde{\mu}_4^{(C_3)}=6$,
so the quartic contribution in Eq.~\eqref{eq:flux_contrast_quartic}
vanishes identically and the first generic flux dependence is delayed
to sixth order.

This cancellation does not survive when the triangle is embedded in the
diamond graph. For a return measurement at the shared vertex $\nu=0$,
one finds
\begin{equation}
\gamma_0=3,
\quad
\widetilde{\mu}^{(0)}_4
=
15
+2\cos\Phi_1
+2\cos\Phi_2
+2\cos(\Phi_1-\Phi_2),
\label{eq:diamond_centered_shared}
\end{equation}
while at the inequivalent outer vertex $\nu=2$,
\begin{equation}
\gamma_2=2,
\qquad
\widetilde{\mu}^{(2)}_4
=
10
-4\cos\Phi_1
+2\cos(\Phi_1-\Phi_2).
\label{eq:diamond_centered_outer}
\end{equation}
The corresponding result for vertex $3$ follows from the exchange
$\Phi_1\leftrightarrow\Phi_2$.

Defining
\begin{equation}
F_4(\alpha)
=
1-\alpha^4-(1-\alpha)^4>0,
\qquad
0<\alpha<1,
\label{eq:F4_definition}
\end{equation}
the general expansion in Eq.~\eqref{eq:flux_contrast_quartic} yields,
relative to the zero-flux dynamics
$\Delta_{\boldsymbol{\Phi}}K_\nu
=
K_{\nu,(\Phi_1,\Phi_2)}
-
K_{\nu,(0,0)}$,
\begin{equation}
\begin{split}
\Delta_{\boldsymbol{\Phi}}K_0(\alpha t,t)
={}&
\frac{F_4(\alpha)}{6}
\Big[
(1-\cos\Phi_1)
+(1-\cos\Phi_2)
\\
&+
\bigl(1-\cos(\Phi_1-\Phi_2)\bigr)
\Big]t^4
+O(t^6),
\end{split}
\label{eq:diamond_K_shared}
\end{equation}
and, for the outer root,
\begin{equation}
\begin{split}
\Delta_{\boldsymbol{\Phi}}K_2(\alpha t,t)
={}&
\frac{F_4(\alpha)}{6}
\Big[
\bigl(1-\cos(\Phi_1-\Phi_2)\bigr)
\\
&-2\bigl(1-\cos\Phi_1\bigr)
\Big]t^4
+O(t^6).
\end{split}
\label{eq:diamond_K_outer}
\end{equation}

Thus, although the diamond graph has the same shortest cycle length as the isolated triangle, its generic flux response is quartic rather than sixth order.
The delayed response of an isolated odd ring is therefore not a
consequence of cycle parity or shortest cycle length alone: it originates from
cancellations in the complete rooted spectral structure of the uniform
cycle, which are lifted by the modified set of rooted closed walks in
the multicycle graph. Equations~\eqref{eq:diamond_K_shared} and
\eqref{eq:diamond_K_outer} also show that both the magnitude and the
sign of the leading flux response may depend on the measured vertex.
Notably, the quartic response does not originate from the even
four-cycle $C_{\mathrm{out}}$ alone: for $\Phi_1=\Phi_2$, the
quadrilateral flux vanishes, yet
\begin{equation}
\Delta_{\boldsymbol{\Phi}}K_0(\alpha t,t)
=
\frac{F_4(\alpha)}{3}
\left(1-\cos\Phi_1\right)t^4
+O(t^6)
\label{eq:diamond_equal_flux_limit}
\end{equation}
remains strictly positive for generic $\Phi_1$, while
$\Delta_{\boldsymbol{\Phi}}K_2$ becomes negative. The quartic onset is
therefore driven by the triangle fluxes themselves, whose isolated-ring
cancellation is removed by the nonuniform degree structure.
As established generally in Corollary~\ref{cor:flux_reversal}, the return protocol is
invariant under simultaneous reversal of all fluxes. It need not,
however, be invariant under reversal of only one independent flux.
For the
representative roots $\nu=0,2$,
\begin{equation}
\begin{split}
&K_{\nu,(-\Phi_1,\Phi_2)}(\alpha t,t)
-
K_{\nu,(\Phi_1,\Phi_2)}(\alpha t,t)
\\
&\hspace{1.3cm}
=
\frac{F_4(\alpha)}{3}
\sin\Phi_1\sin\Phi_2\,t^4
+O(t^6).
\end{split}
\label{eq:diamond_partial_reversal}
\end{equation}
The mixed dependence on $\Phi_1-\Phi_2$ therefore allows local return
statistics to probe the relative configuration of the two cycle fluxes.
This does not provide a directional witness of global
time-reversal-symmetry breaking, since configurations related by
simultaneous reversal remain indistinguishable. Rather, it reflects
interference among rooted closed walks that explore different cycles of
the graph.

\section{Flux-controlled Leggett--Garg correlations}
\label{sec:LGI}
We now use the exact correlator in Eq.~\eqref{eq:exact_correlator} to determine how gauge-invariant fluxes modify the Leggett--Garg functional. Since all two-time correlators are fixed by the rooted return amplitude, the flux dependence of $L_3$ is entirely governed by the way closed-path interference reshapes the local return dynamics.

\paragraph{Equally spaced measurements.}
Following a common choice in three-time Leggett--Garg analyses~\cite{Lambert2014,BudroniEmary2014}, we focus on the equally spaced protocol
$(t_1,t_2,t_3)=(0,\tau,2\tau)$, which provides a one-parameter setting for studying the flux dependence of the temporal correlations.
Since the initial state is the $+1$ eigenstate of $Q_\nu$, the measurement
at $t_1=0$ is deterministic and does not disturb the state. The
correlators involving the initial time therefore reduce to
\begin{align}
C_{12}(\boldsymbol{\Phi})
&=
2p_\nu(\tau;\boldsymbol{\Phi})-1,
\label{eq:C12_equally_spaced}
\\
C_{13}(\boldsymbol{\Phi})
&=
2p_\nu(2\tau;\boldsymbol{\Phi})-1.
\label{eq:C13_equally_spaced}
\end{align}
For the pair $(t_2,t_3)=(\tau,2\tau)$,
Eq.~\eqref{eq:exact_correlator} gives
\begin{equation}
\begin{split}
C_{23}(\boldsymbol{\Phi})
={}&
1
-
2p_\nu(\tau;\boldsymbol{\Phi})
\left[
1-p_\nu(\tau;\boldsymbol{\Phi})
\right]
\\
&-
2\left|
A_\nu(2\tau;\boldsymbol{\Phi})
-
A_\nu(\tau;\boldsymbol{\Phi})^2
\right|^2.
\end{split}
\label{eq:C23_equally_spaced}
\end{equation}

Combining Eqs.~\eqref{eq:C12_equally_spaced}--\eqref{eq:C23_equally_spaced},
the Leggett--Garg functional becomes
\begin{equation}
\begin{split}
L_3(\tau;\boldsymbol{\Phi})
={}&
1
+
2p_\nu(\tau;\boldsymbol{\Phi})^2
-
2p_\nu(2\tau;\boldsymbol{\Phi})
\\
&-
2\left|
A_\nu(2\tau;\boldsymbol{\Phi})
-
A_\nu(\tau;\boldsymbol{\Phi})^2
\right|^2.
\end{split}
\label{eq:L3_equally_spaced}
\end{equation}
The first two dynamical terms describe population recurrences at times
$\tau$ and $2\tau$, whereas the last term contains the interference
between uninterrupted return over the interval $2\tau$ and two consecutive
return amplitudes over intervals of duration $\tau$.
\paragraph{Exact relation to measurement back-action.}

To retain the sign of the measurement-induced change in the final
marginal, we define the signed disturbance
\begin{equation}
\delta_{\nu,\boldsymbol{\Phi}}(s,t)
=
p_\nu^{(s)}(t;\boldsymbol{\Phi})
-
p_\nu(t;\boldsymbol{\Phi}),
\qquad
K_{\nu,\boldsymbol{\Phi}}(s,t)
=
\left|
\delta_{\nu,\boldsymbol{\Phi}}(s,t)
\right|.
\label{eq:signed_disturbance}
\end{equation}
For the equally spaced protocol, Eqs.~\eqref{eq:exact_measured_return}
and~\eqref{eq:L3_equally_spaced} give the exact identity
\begin{equation}
\begin{split}
L_3(\tau;\boldsymbol{\Phi})
={}&
1
+
2\delta_{\nu,\boldsymbol{\Phi}}(\tau,2\tau)
\\
&-
4\left|
A_\nu(2\tau;\boldsymbol{\Phi})
-
A_\nu(\tau;\boldsymbol{\Phi})^2
\right|^2 .
\end{split}
\label{eq:L3_backaction_identity}
\end{equation}
Equivalently, the last term is
$4P_{-+}(\tau,2\tau)$. Since $P_{-+}\geq0$, an LGI violation
requires
\begin{equation}
L_3>1
\quad\Longleftrightarrow\quad
\delta_{\nu,\boldsymbol{\Phi}}(\tau,2\tau)
>
2P_{-+}(\tau,2\tau),
\label{eq:L3_violation_condition}
\end{equation}
and therefore implies nonzero marginal back-action,
$K_{\nu,\boldsymbol{\Phi}}(\tau,2\tau)>0$. Moreover,
\begin{equation}
L_3(\tau;\boldsymbol{\Phi})
\leq
1+
2K_{\nu,\boldsymbol{\Phi}}(\tau,2\tau).
\label{eq:L3_K_bound}
\end{equation}
More generally, the same identity holds for arbitrary measurement times $(0,s,t)$, with $0<s<t$, upon replacing $A_\nu(\tau)^2$ by $A_\nu(t-s)A_\nu(s)$. 

The two diagnostics are nevertheless not redundant.
The quantity $K_\nu$ retains neither the sign of the marginal
disturbance nor the branch-resolved probability $P_{-+}$.
Consequently, nonzero measurement back-action is not sufficient for an
LGI violation: the disturbance must be positive and must exceed the
contribution associated with an intermediate outcome $-1$ followed by
return to the measured vertex. Thus, $K_\nu$ quantifies marginal
measurement disturbance, whereas $L_3$ probes how this disturbance is
distributed among the outcome-conditioned temporal branches.
The implication $L_3>1\Rightarrow K_\nu>0$ relies on the localized
initial preparation, for which the outcome at $t_1=0$ is deterministic,
and is not a universal relation between pairwise NSIT and LGI.

At fixed measurement spacing, we quantify the flux-induced modification
relative to the zero-flux dynamics through
\begin{equation}
\Delta_{\boldsymbol{\Phi}}L_3(\tau)
=
L_3(\tau;\boldsymbol{\Phi})
-
L_3(\tau;\boldsymbol{0}),
\label{eq:flux_L3_contrast}
\end{equation}
where $\boldsymbol{0}$ denotes the configuration in which all independent
cycle fluxes vanish.

\paragraph{Short-time behavior.}

Using the centered moments introduced in
Eqs.~\eqref{eq:centered_definitions}--\eqref{eq:chi_definition}, the return
probability and the interference term admit the expansions
\begin{align}
&p_\nu(\tau;\boldsymbol{\Phi})
=
1
-
\gamma_\nu\tau^2
+
\zeta_\nu(\boldsymbol{\Phi})\tau^4
+
O(\tau^6),
\label{eq:return_probability_short_L3}
\\
&\left|
A_\nu(2\tau;\boldsymbol{\Phi})
-
A_\nu(\tau;\boldsymbol{\Phi})^2
\right|^2
=
\gamma_\nu^2\tau^4
+
O(\tau^6).
\label{eq:interference_short_L3}
\end{align}
Substitution into Eq.~\eqref{eq:L3_equally_spaced} yields
\begin{equation}
L_3(\tau;\boldsymbol{\Phi})
=
1
+
4\gamma_\nu\tau^2
-
28\zeta_\nu(\boldsymbol{\Phi})\tau^4
+
O(\tau^6).
\label{eq:L3_short_time}
\end{equation}
Thus, whenever $\gamma_\nu>0$, the Leggett--Garg inequality is violated
for sufficiently small nonzero $\tau$. The leading quadratic contribution
depends only on the total coupling strength at the measured vertex and is
independent of the Peierls phases.

For fixed hopping magnitudes and on-site energies, $\gamma_\nu$ is
independent of the fluxes. The short-time contrast therefore satisfies
\begin{equation}
\begin{split}
\Delta_{\boldsymbol{\Phi}}L_3(\tau)
&=
-28\,
\Delta_{\boldsymbol{\Phi}}\zeta_\nu\,
\tau^4
+
O(\tau^6)
\\
&=
-\frac{7}{3}\,
\Delta_{\boldsymbol{\Phi}}
\widetilde{\mu}_4^{(\nu)}\,
\tau^4
+
O(\tau^6),
\end{split}
\label{eq:L3_flux_short_time}
\end{equation}
where
$\Delta_{\boldsymbol{\Phi}}
\widetilde{\mu}_4^{(\nu)}
=
\widetilde{\mu}_4^{(\nu)}(\boldsymbol{\Phi})
-
\widetilde{\mu}_4^{(\nu)}(\boldsymbol{0})$.
The fourth order is therefore the earliest order at which flux sensitivity
can appear in the equally spaced Leggett--Garg functional. As for the
measurement back-action, however, cancellations among the complete set of
rooted closed-walk contributions may postpone the first nonzero
flux-dependent term.

\subsection{Balanced rooted Krylov dynamics and saturation of the
L\"uders bound}
\label{subsec:rooted_krylov_saturation}

The short-time expansion in Eq.~\eqref{eq:L3_flux_short_time}
characterizes the onset of the flux dependence at a fixed measurement
spacing. The exact expression in Eq.~\eqref{eq:L3_equally_spaced} also
allows one to identify a complementary, graph-independent mechanism for
maximizing the Leggett--Garg violation. This mechanism is most naturally
expressed in terms of the Krylov subspace rooted at the measured vertex.

For a fixed flux configuration $\boldsymbol{\Phi}$, define
\begin{equation}
\mathscr{K}_{\nu}(\boldsymbol{\Phi})
=
\operatorname{span}
\left\{
\ket{\nu},
H_{\boldsymbol{\Phi}}\ket{\nu},
H_{\boldsymbol{\Phi}}^2\ket{\nu},
\ldots
\right\}.
\label{eq:rooted_Krylov_space}
\end{equation}
Let
$H_{\boldsymbol{\Phi}}
=
\sum_r
\lambda_r(\boldsymbol{\Phi})P_r(\boldsymbol{\Phi})$
be the spectral decomposition of the Hamiltonian, and let
$w_r^{(\nu)}(\boldsymbol{\Phi})
=
\bra{\nu}P_r(\boldsymbol{\Phi})\ket{\nu}$
denote the corresponding rooted spectral weights. The rooted spectral
measure is
\begin{equation}
\mu_{\nu,\boldsymbol{\Phi}}
=
\sum_r
w_r^{(\nu)}(\boldsymbol{\Phi})
\delta_{\lambda_r(\boldsymbol{\Phi})}.
\label{eq:rooted_spectral_measure_Krylov}
\end{equation}
For a finite Hermitian Hamiltonian,
$\dim\mathscr{K}_{\nu}(\boldsymbol{\Phi})$ equals the number of distinct
eigenvalues carrying nonzero rooted spectral weight.

\paragraph{Two-dimensional rooted dynamics.}

Suppose that
$\dim\mathscr{K}_{\nu}(\boldsymbol{\Phi})=2$. The rooted spectral measure is then supported on two distinct eigenvalues
$\lambda_-<\lambda_+$ carrying nonzero spectral weight at the measured
vertex,
\begin{equation}
\mu_{\nu,\boldsymbol{\Phi}}
=
w\,\delta_{\lambda_-}
+
(1-w)\,\delta_{\lambda_+},
\qquad
0<w<1,
\label{eq:two_eigenvalue_rooted_measure}
\end{equation}
where the dependence of $w$ and $\lambda_\pm$ on
$\boldsymbol{\Phi}$ is left implicit. Defining
$\Omega=\lambda_+-\lambda_-$, the return amplitude is
\begin{equation}
A_\nu(t;\boldsymbol{\Phi})
=
e^{-i\lambda_-t}
\left[
w+(1-w)e^{-i\Omega t}
\right].
\label{eq:two_eigenvalue_return_amplitude}
\end{equation}
Substitution into Eq.~\eqref{eq:L3_equally_spaced} gives
\begin{equation}
L_3(\tau;\boldsymbol{\Phi})
=
1
+
8w(1-w)
\left[
\cos(\Omega\tau)
-
\cos^2(\Omega\tau)
\right].
\label{eq:two_eigenvalue_L3}
\end{equation}
Since $x-x^2$ is maximized at $x=1/2$, with maximum value $1/4$,
one obtains
\begin{equation}
\max_{\tau\geq0}
L_3(\tau;\boldsymbol{\Phi})
=
1+2w(1-w)
\leq
\frac{3}{2}.
\label{eq:two_eigenvalue_L3_maximum}
\end{equation}
Thus, within the class of two-dimensional rooted Krylov dynamics, the
L\"uders bound is attained if and only if the two rooted spectral weights
are equal, $w=1-w=1/2$. In that case, the first positive saturation time
is
\begin{equation}
\tau_\star
=
\frac{\pi}{3\Omega}.
\label{eq:two_eigenvalue_saturation_time}
\end{equation}

Equation~\eqref{eq:two_eigenvalue_L3_maximum} also shows that a rooted
Krylov dimension equal to two is not sufficient by itself to saturate the
bound. The relative weights with which the two eigenvalues are sampled at
the measured vertex are equally essential.
This mechanism differs from the multilevel enhancement discussed in
Ref.~\cite{BudroniEmary2014}, where a finer projective resolution of the
measurement can raise the quantum bound above the L\"uders value. Here the
measurement remains the fixed dichotomic instrument
$\{\Pi_\nu,\mathbb{I}-\Pi_\nu\}$: gauge flux only reshapes the rooted
dynamics, allowing the bound $L_3=3/2$ to be saturated but not exceeded.

\paragraph{Operator criterion for balanced rooted dynamics.}

Let
$a_\nu=\bra{\nu}H_{\boldsymbol{\Phi}}\ket{\nu}
=\mu_1^{(\nu)}(\boldsymbol{\Phi})$
and recall that
$\gamma_\nu=\widetilde{\mu}_2^{(\nu)}(\boldsymbol{\Phi})$
is the local coupling strength introduced in
Eq.~\eqref{eq:local_coupling_strength}. The dependence of
$a_\nu$ and $\gamma_\nu$ on $\boldsymbol{\Phi}$ is left implicit below.

For $\gamma_\nu>0$, balanced two-dimensional rooted dynamics is
equivalent to
\begin{equation}
\left(
H_{\boldsymbol{\Phi}}-a_\nu\mathbb{I}
\right)^2\ket{\nu}
=
\gamma_\nu\ket{\nu}.
\label{eq:balanced_Krylov_condition}
\end{equation}
Indeed, defining
$\ket{u_1}=
(H_{\boldsymbol{\Phi}}-a_\nu\mathbb{I})
\ket{\nu}/\sqrt{\gamma_\nu}$,
one has
$H_{\boldsymbol{\Phi}}\ket{\nu}
=a_\nu\ket{\nu}+\sqrt{\gamma_\nu}\ket{u_1}$, while
Eq.~\eqref{eq:balanced_Krylov_condition} gives
$H_{\boldsymbol{\Phi}}\ket{u_1}
=\sqrt{\gamma_\nu}\ket{\nu}+a_\nu\ket{u_1}$.
Thus,
$\mathscr{K}_\nu=
\operatorname{span}\{\ket{\nu},\ket{u_1}\}$
is invariant under $H_{\boldsymbol{\Phi}}$, and
\begin{equation}
 H_{\boldsymbol{\Phi}}|_{\mathscr{K}_\nu}
=
\begin{pmatrix}
a_\nu & \sqrt{\gamma_\nu}\\
\sqrt{\gamma_\nu} & a_\nu
\end{pmatrix}.
\label{eq:balanced_Krylov_matrix}
\end{equation}
Its eigenvalues are $a_\nu\pm\sqrt{\gamma_\nu}$, and
$\ket{\nu}$ carries rooted spectral weight $1/2$ on each corresponding
eigenspace. Hence,
\begin{equation}
\mu_{\nu,\boldsymbol{\Phi}}
=
\frac{1}{2}\delta_{a_\nu-\sqrt{\gamma_\nu}}
+
\frac{1}{2}\delta_{a_\nu+\sqrt{\gamma_\nu}}.
\label{eq:balanced_Krylov_measure}
\end{equation}
It follows that
$\max_{\tau\geq0}L_3(\tau;\boldsymbol{\Phi})=3/2$, with first positive
saturation time
$\tau_\star=\pi/(6\sqrt{\gamma_\nu})$.
Conversely, equal rooted weights on two eigenvalues $\lambda_\pm$ imply
$a_\nu=(\lambda_++\lambda_-)/2$ and
$\gamma_\nu=(\lambda_+-\lambda_-)^2/4$, so that
$(H_{\boldsymbol{\Phi}}-a_\nu\mathbb{I})^2
=\gamma_\nu\mathbb{I}$ on the rooted Krylov subspace and
Eq.~\eqref{eq:balanced_Krylov_condition} follows.

Equation~\eqref{eq:balanced_Krylov_condition} also provides a
constructive phase-design criterion. Define the rooted leakage functional
\begin{equation}
\mathcal{R}_\nu(\boldsymbol{\Phi})
=
\left\|
\left(
\mathbb{I}-\ket{\nu}\bra{\nu}
\right)
\left[
H_{\boldsymbol{\Phi}}
-a_\nu(\boldsymbol{\Phi})\mathbb{I}
\right]^2
\ket{\nu}
\right\|^2.
\label{eq:rooted_leakage_functional}
\end{equation}
The conditions
$\mathcal{R}_\nu(\boldsymbol{\Phi})=0$ and
$\gamma_\nu(\boldsymbol{\Phi})>0$
are equivalent to Eq.~\eqref{eq:balanced_Krylov_condition}.
Geometrically, the fluxes enforce destructive interference among the
rooted propagation channels generated by two applications of the
centered Hamiltonian, eliminating all leakage outside a balanced
two-dimensional rooted subspace.

This graph-independent criterion guarantees finite-time saturation of
the L\"uders bound and is necessary and sufficient within the class
$\dim\mathscr{K}_\nu=2$. It is not necessary in full generality, since
Hamiltonians with more than two eigenvalues carrying nonzero rooted
spectral weight may also attain, or approach, the bound through special
relations among their dynamical phases.

\paragraph{Flux-engineered saturation on the diamond graph.}

The diamond graph discussed above provides an exact multicycle
realization of the balanced rooted-Krylov mechanism. For a return
measurement at the outer vertex $\nu=2$ and the relative-flux
configuration $(\Phi_1,\Phi_2)=(0,\pi)$, the Hamiltonian
satisfies
\begin{equation}
\left[
H_{\diamond}(0,\pi)-2\mathbb{I}
\right]^2
\ket{2}
=
2\ket{2}.
\label{eq:diamond_balanced_Krylov}
\end{equation}
Consequently, $\dim\mathscr{K}_2(0,\pi)=2$, and the rooted spectral
measure is
\begin{equation}
\mu_{2,(0,\pi)}
=
\frac{1}{2}\delta_{2-\sqrt{2}}
+
\frac{1}{2}\delta_{2+\sqrt{2}}.
\label{eq:diamond_balanced_measure}
\end{equation}
The remaining eigenmodes have zero rooted weight at vertex $2$ and are
therefore dark with respect to the local return dynamics. Accordingly,
$A_2(t;0,\pi)=e^{-2it}\cos(\sqrt{2}\,t)$, and
Eqs.~\eqref{eq:two_eigenvalue_L3_maximum}
and~\eqref{eq:two_eigenvalue_saturation_time} give
\begin{equation}
\max_{\tau\geq0}
L_{3,2}(\tau;0,\pi)
=
\frac{3}{2},
\qquad
\tau_\star
=
\frac{\pi}{6\sqrt{2}}.
\label{eq:diamond_Luders_saturation}
\end{equation}

At zero flux, the rooted spectral measure at the same vertex is instead
\begin{equation}
\mu_{2,(0,0)}
=
\frac{1}{4}\delta_0
+
\frac{1}{2}\delta_2
+
\frac{1}{4}\delta_4,
\label{eq:diamond_zero_flux_measure}
\end{equation}
corresponding to
$A_2(t;0,0)=e^{-2it}\cos^2 t$. Direct maximization of
Eq.~\eqref{eq:L3_equally_spaced} yields
\begin{equation}
\max_{\tau\geq0}
L_{3,2}(\tau;0,0)
=
\frac{119}{96}
+
\frac{11\sqrt{33}}{288}
\simeq
1.4590.
\label{eq:diamond_zero_flux_L3_maximum}
\end{equation}
The relative flux therefore produces an exact enhancement of
$(75-11\sqrt{33})/288\simeq0.0410$. Thus, it does not merely relocate
the optimal measurement time: it changes the rooted spectral support,
makes two modes dark, and increases the unrestricted maximal violation
to the L\"uders bound.

\paragraph{Maximal strength and temporal accessibility.}

The preceding criterion concerns the unrestricted maximal violation and,
when Eq.~\eqref{eq:balanced_Krylov_condition} is satisfied, guarantees
saturation at the finite time given in
Eq.~\eqref{eq:two_eigenvalue_saturation_time}. In general, however, the
zero-flux and finite-flux dynamics may attain their largest experimentally
accessible violations at substantially different times. To distinguish
the maximal strength of the temporal correlations from their temporal
accessibility, we retain the finite-window optimized quantity
\begin{equation}
L_3^{\star}(\boldsymbol{\Phi};T)
=
\max_{0\leq\tau\leq T}
L_3(\tau;\boldsymbol{\Phi}),
\label{eq:L3_optimized_window}
\end{equation}
where $T$ is the temporal observation window.

The exact correlator identities, the short-time expansion, and the
balanced rooted-Krylov criterion apply to arbitrary finite chiral CTQWs.
The diamond graph provides a concrete multicycle realization of
flux-engineered L\"uders-bound saturation. In the next section, we
specialize the general framework to flux-threaded cycles, for which
the return amplitude can be evaluated exactly and the dependence of
$L_3(\tau;\Phi)$ on the flux, system size, and cycle parity can be studied
explicitly.

\section{Flux-threaded cycles}
\label{sec:cycle}

We now apply the general framework to the cycle graph $C_N$ threaded by a uniform gauge-invariant flux $\phi$. Translation invariance makes the return dynamics independent of the measured vertex, while the single-cycle geometry permits an exact winding-number representation.

\subsection{Spectrum and winding-number expansion}
\label{subsec:cycle_spectrum}

Labeling the vertices by $x=0,\ldots,N-1$, with indices understood modulo $N$, the Hamiltonian is
\begin{equation}
H_\phi
=
2\mathbb{I}
-
\sum_{x=0}^{N-1}
\left(
e^{i\phi/N}\ket{x+1}\!\bra{x}
+
e^{-i\phi/N}\ket{x}\!\bra{x+1}
\right).
\label{eq:cycle_hamiltonian}
\end{equation}
The additive term $2\mathbb{I}$ contributes only a global dynamical phase and therefore does not affect probabilities or sequential correlations.

In the uniform gauge, the eigenvectors are the Fourier modes
\begin{equation}
\ket{k}
=
\frac{1}{\sqrt{N}}
\sum_{x=0}^{N-1}
e^{iq_kx}\ket{x},
\qquad
q_k=\frac{2\pi k}{N},
\label{eq:cycle_fourier_modes}
\end{equation}
with $k=0,\ldots,N-1$, and the corresponding eigenvalues are
\begin{equation}
\lambda_k(\phi)
=
2-2\cos\left(q_k-\frac{\phi}{N}\right).
\label{eq:cycle_eigenvalues}
\end{equation}
The flux thus shifts the dispersion relation relative to the discrete quasimomentum grid.

Since all vertices are equivalent, the return amplitude is independent of the measured vertex:
\begin{equation}
A_N(t;\phi)
=
\frac{1}{N}
\sum_{k=0}^{N-1}
e^{-i\lambda_k(\phi)t},
\label{eq:cycle_return_amplitude}
\end{equation}
with return probability $p_N(t;\phi)=|A_N(t;\phi)|^2$.

Using the Jacobi--Anger expansion, Eq.~\eqref{eq:cycle_return_amplitude} becomes
\begin{align}
A_N(t;\phi)
&=
\frac{e^{-2it}}{N}
\sum_{k=0}^{N-1}
\exp\left[
2it\cos\left(q_k-\frac{\phi}{N}\right)
\right]
\nonumber\\
&=
e^{-2it}
\sum_{\ell\in\mathbb{Z}}
i^{\ell N}J_{\ell N}(2t)e^{-i\ell\phi},
\label{eq:cycle_winding_expansion}
\end{align}
where $J_n$ is the Bessel function of the first kind. Combining opposite winding numbers gives
\begin{equation}
A_N(t;\phi)
=
e^{-2it}
\left[
J_0(2t)
+
2\sum_{\ell=1}^{\infty}
i^{\ell N}J_{\ell N}(2t)\cos(\ell\phi)
\right].
\label{eq:cycle_winding_cosine}
\end{equation}

The integer $\ell$ is the net winding number of a closed path around the cycle. More precisely, the sectors $\ell$ reorganize the rooted closed walks of Appendix~\ref{app:closed-walks} according to their winding number $n_C(w)=\ell$. The even--odd onset derived below is therefore the cycle-specific manifestation of the general closed-walk cancellation mechanism discussed in Sec.~\ref{sec:flux-onset}.

Equation~\eqref{eq:cycle_winding_cosine} also implies
\begin{align}
A_N(t;\phi+2\pi)
&=
A_N(t;\phi),
\label{eq:cycle_flux_periodicity}\\
A_N(t;-\phi)
&=
A_N(t;\phi).
\label{eq:cycle_flux_evenness}
\end{align}
Consequently, all return-based quantities considered here, including $p_N$, $K_{N,\phi}$, and $L_{3,N}$, are even and $2\pi$-periodic functions of the flux. For odd $N$, there is an additional half-flux symmetry. Reindexing the spectrum gives $\lambda_{k+(N+1)/2}(\phi+\pi)=4-\lambda_k(\phi)$,
and hence $A_N(t;\phi+\pi)=e^{-4it}A_N(t;\phi)^*$. Therefore all return-based quantities considered here are $\pi$-periodic for odd cycles. Combined with flux-reversal symmetry, this also implies invariance under $\phi\mapsto\pi-\phi$, and in particular exact equivalence between $\phi=0$ and $\phi=\pi$.

The same expansion provides an exact harmonic decomposition at arbitrary times. Each harmonic $\cos(\ell\phi)$ corresponds to a winding sector $\ell$, while products of return amplitudes generate sum- and difference-frequency harmonics in $p_N$, $L_{3,N}$, and the signed expression entering $K_{N,\phi}$. Moreover, $J_{\ell N}(2t)$ becomes appreciable near its turning region $2t\sim|\ell|N$, defining the characteristic scale
\begin{equation}
t_\ell\sim\frac{|\ell|N}{2}.
\label{eq:winding_time_scale}
\end{equation}
Since the maximal group velocity is $v_{\max}=2$, this is the time required to traverse approximately $|\ell|$ circumferences of the cycle. Flux sensitivity therefore exhibits a crossover rather than a strict temporal threshold. The short-time powers derived below reflect both the suppression of finite-winding sectors and parity-dependent cancellations in return-based observables.

\subsection{Even--odd onset of flux sensitivity}
\label{subsec:cycle_flux_onset}

We set $s=\alpha t$, with $0<\alpha<1$, and denote the cycle specialization of the contrast in Eq.~\eqref{eq:flux_K_contrast} by $\Delta_\phi K_N(\alpha t,t)$. For later convenience, define
\begin{align}
F_N(\alpha)
&=
1-\alpha^N-(1-\alpha)^N,
\label{eq:F_N_alpha}\\
G_N(\alpha)
&=
1-\alpha^{2N}-(1-\alpha)^{2N}.
\label{eq:G_N_alpha}
\end{align}
Both functions are positive for $0<\alpha<1$.

\begin{proposition}[Parity-dependent flux onset]
\label{prop:cycle_parity_onset}
For the flux-threaded cycle $C_N$, with $N\geq3$, the short-time flux contrast of the return-based Kolmogorov inconsistency satisfies the following relations.

For even $N$,
\begin{equation}
\Delta_\phi K_N(\alpha t,t)
=
\frac{4(-1)^{N/2}}{N!}
\left(1-\cos\phi\right)
F_N(\alpha)t^N
+
O(t^{N+2}).
\label{eq:even_cycle_flux_onset}
\end{equation}

For odd $N$,
\begin{align}
\Delta_\phi K_N(\alpha t,t)
={}&
\sin^2\phi
\Bigg[
\left(
\frac{4}{(N!)^2}
-
\frac{8}{(2N)!}
\right)
G_N(\alpha)
\nonumber\\
&
-
\frac{4}{(N!)^2}
F_N(\alpha)^2
\Bigg]
t^{2N}
+
O(t^{2N+2}).
\label{eq:odd_cycle_flux_onset}
\end{align}
\end{proposition}

\begin{proof}
The small-argument expansion
\begin{equation}
J_m(2t)=\frac{t^m}{m!}+O(t^{m+2})
\end{equation}
shows that the sectors $\ell=\pm1$ generate the first flux-dependent contribution to the return amplitude. This term is proportional to
$i^NJ_N(2t)\cos\phi$ and is therefore of order $t^N$.

For even $N$, $i^N=(-1)^{N/2}$ is real. The first winding contribution therefore interferes directly with the zero-winding amplitude and produces a flux-dependent term at order $t^N$. Substitution into Eq.~\eqref{eq:exact_K} yields Eq.~\eqref{eq:even_cycle_flux_onset}.

For odd $N$, $i^N$ is purely imaginary. Its linear interference with the real zero-winding contribution vanishes in return probabilities and in the return-based Kolmogorov inconsistency. The leading flux dependence then occurs at order $t^{2N}$ through the squared single-winding contribution and the $\ell=\pm2$ winding sectors. Their combination yields Eq.~\eqref{eq:odd_cycle_flux_onset}. A derivation based on the winding-number expansion is given in Appendix \ref{app:cycle-asymptotics}.
\end{proof}

The winding contribution therefore enters the amplitude at order $t^N$ for every cycle but becomes observable at that order only for even $N$. For even cycles, the sign of the leading correction alternates with $N/2$: at sufficiently short times, a generic nonzero flux enhances the inconsistency for $N=0\pmod 4$ and suppresses it for $N=2\pmod 4$. For odd cycles, the response is delayed to order $t^{2N}$.

\subsection{Finite-time measurement back-action}
\label{subsec:cycle_backaction}

At arbitrary times, the back-action follows by substituting Eq.~\eqref{eq:cycle_return_amplitude} into the exact expression~\eqref{eq:exact_K}. For finite-time comparisons, we use the midpoint protocol
\begin{equation}
K_N^{\mathrm{mid}}(t;\phi)
=
K_{N,\phi}\left(\frac{t}{2},t\right),
\label{eq:cycle_midpoint_K}
\end{equation}
and define its contrast relative to zero flux as
\begin{equation}
\Delta_\phi K_N^{\mathrm{mid}}(t)
=
K_N^{\mathrm{mid}}(t;\phi)
-
K_N^{\mathrm{mid}}(t;0).
\label{eq:cycle_midpoint_contrast}
\end{equation}
The choice $\alpha=1/2$ maximizes the universal quadratic contribution and, for even cycles, the factor $F_N(\alpha)$ controlling the leading flux contrast.

\begin{figure}[t]
\centering
\includegraphics[width=\columnwidth]{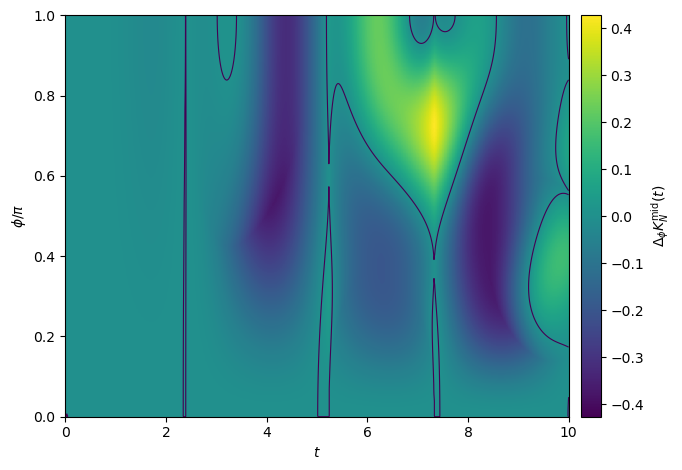}
\caption{
Finite-time flux control of measurement back-action on the cycle $C_6$. The color scale represents the midpoint contrast
$\Delta_{\phi}K_6^{\mathrm{mid}}(t)
=
K_{6,\phi}(t/2,t)-K_{6,0}(t/2,t)$
as a function of $t$ and $\phi/\pi$. Positive and negative values indicate enhancement and suppression relative to zero flux, respectively. The black curves mark
$\Delta_{\phi}K_6^{\mathrm{mid}}(t)=0$.
}
\label{fig:finite-time-backaction}
\end{figure}

For $C_6$, Fig.~\ref{fig:finite-time-backaction} shows the initial suppression predicted by the factor $(-1)^{N/2}$ in Eq.~\eqref{eq:even_cycle_flux_onset}. At longer times, interference among several winding sectors generates alternating regions of enhancement and suppression. The flux therefore reorganizes the temporal structure of the measurement disturbance rather than merely rescaling its magnitude.

\subsection{Flux-controlled Leggett--Garg correlations}
\label{subsec:flux_lg}

For the equally spaced protocol $(0,\tau,2\tau)$, let $L_{3,N}(\tau;\phi)$ denote the functional in Eq.~\eqref{eq:L3_equally_spaced} evaluated with the cycle return amplitude $A_N(t;\phi)$. The return-probability terms encode the flux-dependent recurrence structure, while the interference term compares direct return over $2\tau$ with two consecutive return amplitudes over intervals of duration $\tau$.

Equations~\eqref{eq:cycle_flux_periodicity} and \eqref{eq:cycle_flux_evenness} imply
\begin{align}
L_{3,N}(\tau;\phi+2\pi)
&=
L_{3,N}(\tau;\phi),
\nonumber\\
L_{3,N}(\tau;-\phi)
&=
L_{3,N}(\tau;\phi).
\label{eq:symmetries}
\end{align}
It is therefore sufficient to consider $0\leq\phi\leq\pi$. The points $\phi=0$ and $\phi=\pi$ are invariant under time reversal, whereas generic fluxes break time-reversal symmetry. A Leggett--Garg violation occurs whenever $L_{3,N}(\tau;\phi)>1$.

For an observation window $0\leq\tau\leq T$, define
\begin{equation}
V_{3,N}^{\star}(\phi;T)
=
\left[
L_{3,N}^{\star}(\phi;T)-1
\right]_{+},
\qquad
[x]_{+}=\max\{x,0\},
\label{eq:excess}
\end{equation}
where $L_{3,N}^{\star}$ is the finite-window maximum introduced in Eq.~\eqref{eq:L3_optimized_window}. Since $L_{3,N}(0;\phi)=1$, one has $V_{3,N}^{\star}>0$ if and only if a violation occurs within the chosen window.

\paragraph{Analytic benchmark: $C_4$ at half flux.}

At half flux, the rooted spectral measure of $C_4$ coincides with that of the balanced outer vertex of the diamond graph in Eq.~\eqref{eq:diamond_balanced_measure}. Consequently,
$A_4(t;\pi)
=
e^{-2it}\cos(\sqrt{2}\,t)$,
and
\begin{equation}
\max_{\tau\geq0}L_{3,4}(\tau;\pi)
=
\frac{3}{2},
\qquad
\tau_\star
=
\frac{\pi}{6\sqrt{2}}.
\label{eq:C4_half_flux_saturation}
\end{equation}
At zero flux, Eq.~\eqref{eq:diamond_zero_flux_measure} gives
$A_4(t;0)
=
e^{-2it}\cos^2t$,
and
\begin{equation}
\max_{\tau\geq0}L_{3,4}(\tau;0)
=
\frac{119}{96}
+
\frac{11\sqrt{33}}{288}
\simeq1.4590.
\label{eq:C4_zero_flux_maximum}
\end{equation}
Thus, within the equally spaced protocol, half flux increases the maximal violation to the L\"uders bound rather than merely shifting the optimal measurement time. The correspondence with the diamond graph also shows that the sequential return statistics probe rooted spectral equivalence rather than graph topology alone.

For $C_6$ and $C_{10}$, the optimized violation varies by less than one percent over the flux range in the observation windows considered below. The dominant effect is instead temporal: flux can move a prescribed violation to an earlier recurrence. To quantify this, assume $V_{3,N}^{\star}(0;T)>0$ and fix a fraction $0<\eta\leq1$ of the optimized zero-flux violation excess. The target value is
\begin{equation}
L_{\mathrm{target}}^{(N)}(\eta;T)
=
1+\eta V_{3,N}^{\star}(0;T).
\label{eq:target}
\end{equation}
The earliest time at which this target is reached at flux $\phi$ is
\begin{equation}
\tau_{\mathrm{hit}}^{(N)}(\phi;\eta,T)
=
\inf
\left\{
\tau\in(0,T]:
L_{3,N}(\tau;\phi)
\geq
L_{\mathrm{target}}^{(N)}(\eta;T)
\right\},
\label{eq:hitting_time}
\end{equation}
with the convention $\inf\varnothing=+\infty$. The corresponding temporal speedup is
\begin{equation}
S_N(\eta;T)
=
\frac{
\tau_{\mathrm{hit}}^{(N)}(0;\eta,T)
}{
\displaystyle
\inf_{0\leq\phi\leq\pi}
\tau_{\mathrm{hit}}^{(N)}(\phi;\eta,T)
}.
\label{eq:speedup}
\end{equation}
Because the optimization interval includes $\phi=0$, one has $S_N\geq1$. Values $S_N>1$ indicate earlier access to the same prescribed violation, without necessarily implying a larger optimized value.

\begin{figure}[t]
\centering
\includegraphics[width=\columnwidth]{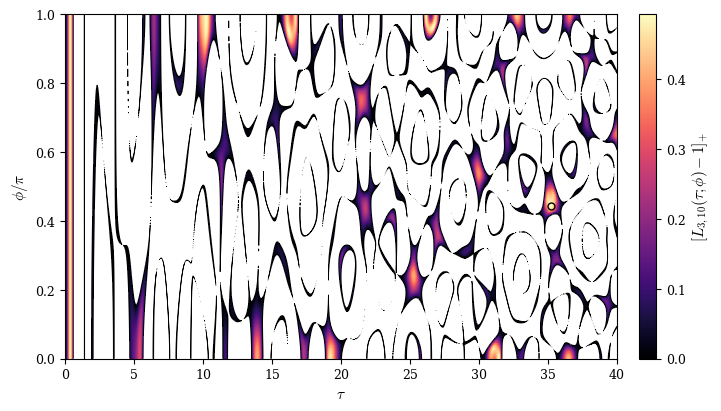}
\caption{
Leggett--Garg violation regions for the cycle $C_{10}$. The color scale represents the positive violation excess
$[L_{3,10}(\tau;\phi)-1]_{+}$
as a function of $\tau$ and $\phi/\pi$. White regions do not violate the macrorealist bound, while the black curves mark
$L_{3,10}(\tau;\phi)=1$. The largest value in the displayed window is
$L_{3,10}\simeq1.4931$, attained near
$\tau\simeq35.2$ and $\phi/\pi\simeq0.445$.
}
\label{fig:L3-phase-diagram-C10}
\end{figure}

Figure~\ref{fig:L3-phase-diagram-C10} shows that the flux changes both the location and temporal extent of the violation regions. The largest value in the displayed window remains below the L\"uders bound and differs only weakly from the zero-flux optimum.

\begin{figure*}[t]
\centering
\includegraphics[width=0.90\textwidth]{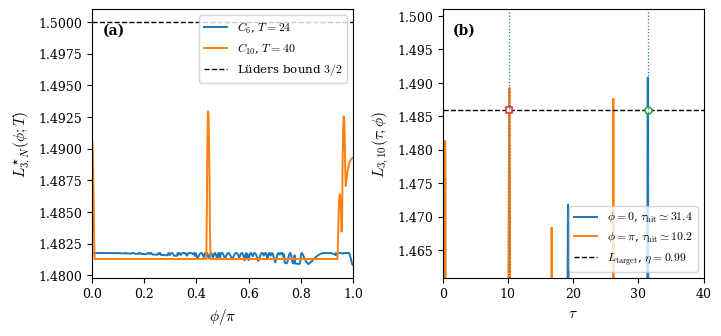}
\caption{
Flux dependence of the Leggett--Garg violation strength and its temporal accessibility for $C_6$ and $C_{10}$.
\textbf{(a)} Optimized functional
$L_{3,N}^{\star}(\phi;T)$
for $C_6$ with $T=24$ and $C_{10}$ with $T=40$. The dashed line marks the L\"uders bound $L_3=3/2$.
\textbf{(b)} Leggett--Garg functional for $C_{10}$ at zero and half flux. The horizontal dashed line marks
$L_{\mathrm{target}}^{(10)}(0.99;40)
=
1+0.99V_{3,10}^{\star}(0;40)$.
The first target crossings occur at
$\tau_{\mathrm{hit}}^{(10)}(0;0.99,40)\simeq31.4$
and
$\tau_{\mathrm{hit}}^{(10)}(\pi;0.99,40)\simeq10.2$,
giving
$S_{10}(0.99;40)\simeq3.1$.
}
\label{fig:L3-strength-and-timing}
\end{figure*}

As shown in Fig.~\ref{fig:L3-strength-and-timing}(a), the optimized violations for $C_6$ and $C_{10}$ depend only weakly on flux. The stronger effect concerns temporal accessibility. For $C_{10}$, $T=40$, and $\eta=0.99$, Fig.~\ref{fig:L3-strength-and-timing}(b) shows that half flux relocates the target from a late zero-flux recurrence at $\tau\simeq31.4$ to an earlier interference peak at $\tau\simeq10.2$, corresponding to
\begin{equation}
S_{10}(0.99;40)\simeq3.1.
\end{equation}

A numerical survey over
$4\leq N\leq16$,
$2N\leq T\leq10N$, and
$0.95\leq\eta\leq0.999$
identifies finite parameter regions with $S_N>1$ for
$N=4,7,8,10,$ and $14$. Speedups of order $2.5$--$3$ occur over multiple sampled thresholds or observation windows for $N=8,10,$ and $14$. For $N\geq6$, the largest advantages typically arise when the target lies above the first nearly flux-insensitive short-time peak, so that the zero-flux trajectory reaches the target only at a later recurrence.

For the representative $C_{10}$ example, the optimal flux is $\phi=\pi$, which is itself time-reversal invariant. The speedup must therefore be attributed to a flux-induced rearrangement of spectral gaps and recurrences rather than to time-reversal-symmetry breaking. Together with the $C_4$ result, this distinguishes flux control of the maximal violation from flux control of its temporal accessibility. Quantifying the possible advantage of earlier violations in the presence of decoherence requires an explicit open-system analysis.

\subsection{Rooted Krylov structure}
\label{subsec:cycle_krylov}

Each Fourier mode has the same local weight,
\begin{equation}
\left|\langle\nu|k\rangle\right|^2
=
\frac{1}{N},
\qquad
\forall\,\nu,k.
\label{eq:cycle_local_spectral_weights}
\end{equation}
An eigenspace of degeneracy $g$ therefore carries total rooted spectral weight $g/N$. The rooted Krylov dimension
\begin{equation}
m_N(\phi)
=
\dim\mathscr{K}_{\nu}(\phi)
\end{equation}
equals the number of distinct eigenvalues in Eq.~\eqref{eq:cycle_eigenvalues} and, by translation invariance, is independent of the root.

At zero flux, the spectral pairing $k\leftrightarrow N-k$ gives
\begin{equation}
m_N(0)
=
\begin{cases}
\dfrac{N+1}{2}, & N\ \text{odd},\\[2mm]
\dfrac{N}{2}+1, & N\ \text{even}.
\end{cases}
\label{eq:krylov_zero_flux}
\end{equation}
At half flux, the pairing becomes $k\leftrightarrow1-k$ modulo $N$, yielding
\begin{equation}
m_N(\pi)
=
\begin{cases}
\dfrac{N+1}{2}, & N\ \text{odd},\\[2mm]
\dfrac{N}{2}, & N\ \text{even}.
\end{cases}
\label{eq:krylov_half_flux}
\end{equation}
For generic flux, the momentum-pair degeneracies are lifted:
\begin{equation}
m_N(\phi)=N,
\qquad
\phi\not\equiv0,\pi\pmod{2\pi}.
\label{eq:krylov_generic_flux}
\end{equation}
Thus, the flux changes the multiplicities, spectral gaps, and number of distinct frequencies contributing to the rooted dynamics, but not the Hilbert-space dimension or the individual Fourier-mode weights.

These counts single out $C_4$ at half flux as the unique simple cycle whose rooted spectral measure is supported on two distinct eigenvalues with equal weight. This explains its saturation of the L\"uders bound through the balanced rooted-Krylov mechanism.

The rooted Krylov dimension alone does not, however, determine either violation strength or temporal accessibility. For $C_{10}$,
\begin{equation}
m_{10}(\pi)=5,
\qquad
m_{10}(0)=6,
\qquad
m_{10}(\phi_{\mathrm{generic}})=10,
\end{equation}
yet the earliest strong violation occurs at half flux. This behavior depends on the complete rooted spectral measure---including multiplicities, weights, gaps, and frequency commensurabilities---rather than on the number of distinct rooted eigenvalues alone.

\section{Conclusions and outlook}
\label{sec:conclusions}

We have shown that single-vertex sequential return measurements provide local probes of gauge-invariant closed-path interference in chiral continuous-time quantum walks. For a walker initially localized at the measured vertex, the complete two-time statistics generated by the dichotomic observable are determined by the rooted return amplitude. The exact relation between the two diagnostics yields $L_3\leq1+2K_{\nu,\Phi}$, implying that any Leggett--Garg violation in the present protocol necessarily entails nonzero measurement back-action. The protocol therefore uses a single-vertex dichotomic readout while retaining information about the rooted spectral structure of the full graph.

At short times, the leading measurement disturbance is independent of the Peierls phases, whereas flux sensitivity enters through gauge-invariant sums of rooted closed walks and may be delayed by interference cancellations. For flux-threaded cycles, the generic leading contrast occurs at order $t^N$ for even $N$ and at order $t^{2N}$ for odd $N$. The diamond graph shows that this delay is not fixed by parity or shortest cycle length alone: in a multicycle network, the response depends on the measured vertex and on relative combinations of the independent fluxes. Although the protocol is invariant under simultaneous reversal of all fluxes, it can therefore distinguish inequivalent relative-flux configurations.

We also derived a graph-independent sufficient criterion for saturating the L\"uders bound. When the rooted dynamics reduces to a balanced two-dimensional Krylov subspace, the maximum over the measurement spacing reaches $L_3=3/2$. Gauge flux can realize this balanced reduction through distinct spectral mechanisms. On the diamond graph, destructive interference renders additional rooted modes dark, whereas for $C_4$ at half flux, flux-induced degeneracies collapse the rooted spectral
support to two equally weighted eigenvalues. For the larger cycles examined, flux has only a weak effect on the optimized violation but can substantially advance the time at which a strong violation is reached; for the representative $C_{10}$ case, the resulting temporal speedup is approximately $3.1$. Gauge flux can thus control either the maximal strength or the temporal accessibility of quantum temporal correlations.

Natural extensions include studying the robustness of these effects under decoherence, which can constrain the maximal observable Leggett--Garg
violation \cite{Emary2013}. In particular, it would be interesting to determine whether flux-induced saturation or earlier accessibility enlarges the regime in which Leggett--Garg violations remain observable. A complementary direction is the metrological characterization of flux inference from local return statistics. Recent work has connected Leggett--Garg violations to quantum Fisher information in stationary settings~\cite{Abboud2026}. Whether analogous connections hold for the nonstationary localized preparation considered here and for flux estimation remains an open question. Comparing the quantum Fisher information with the classical Fisher information of the dichotomic sequential protocol could quantify its efficiency and, in multicycle networks, identify which individual or relative flux combinations are locally estimable. This may enable sparse probes of synthetic gauge fields without complete position-resolved detection.

\appendix


\section{Derivation of the exact sequential return statistics}
\label{app:exact-statistics}

We derive the exact two-time statistics of the return observable
$Q_\nu=2\Pi_\nu-\mathbb{I}$ for a walker initially localized at the
measured vertex, $\rho_0=\Pi_\nu$, where
$\Pi_\nu=\ket{\nu}\!\bra{\nu}$ and
$\Pi_\nu^\perp=\mathbb{I}-\Pi_\nu$. We denote the propagator by
$U_t=e^{-iHt}$, the return amplitude at the measured vertex by $A_\nu(t)=\bra{\nu}U_t\ket{\nu}$, and the corresponding return probability by
$p_\nu(t)=|A_\nu(t)|^2$. The dependence on the gauge-invariant fluxes is
suppressed throughout this appendix.

Consider two measurements at times $t_1<t_2$, and let
$\tau=t_2-t_1$. The probabilities for the sequences $(+,+)$ and
$(+,-)$ are
\begin{align}
P_{++}(t_1,t_2)
&=
\left\|
\Pi_\nu U_\tau\Pi_\nu U_{t_1}\ket{\nu}
\right\|^2
=
p_\nu(t_1)p_\nu(\tau),
\\
P_{+-}(t_1,t_2)
&=
\left\|
\Pi_\nu^\perp U_\tau\Pi_\nu U_{t_1}\ket{\nu}
\right\|^2
=
p_\nu(t_1)\left[1-p_\nu(\tau)\right].
\end{align}

For the sequence $(-,+)$, the relevant transition amplitude is
\begin{align}
\bra{\nu}U_\tau\Pi_\nu^\perp U_{t_1}\ket{\nu}
&=
\bra{\nu}U_\tau U_{t_1}\ket{\nu}
-
\bra{\nu}U_\tau\Pi_\nu U_{t_1}\ket{\nu}
\nonumber
\\
&=
A_\nu(t_2)-A_\nu(\tau)A_\nu(t_1).
\end{align}
It follows that
\begin{equation}
P_{-+}(t_1,t_2)
=
\left|
A_\nu(t_2)-A_\nu(\tau)A_\nu(t_1)
\right|^2.
\end{equation}
Finally, using
$P_{-+}+P_{--}=1-p_\nu(t_1)$ gives
\begin{equation}
P_{--}(t_1,t_2)
=
1-p_\nu(t_1)
-
\left|
A_\nu(t_2)-A_\nu(\tau)A_\nu(t_1)
\right|^2.
\end{equation}

For a nonselective intermediate measurement at time $s<t$, the
probability of finding the walker at $\nu$ at the final time is obtained
by summing over the intermediate outcome:
\begin{align}
p_\nu^{(s)}(t)
&=
P_{++}(s,t)+P_{-+}(s,t)
\nonumber
\\
&=
p_\nu(s)p_\nu(t-s)
+
\left|
A_\nu(t)-A_\nu(t-s)A_\nu(s)
\right|^2.
\label{eq:app-exact-measured-return}
\end{align}
Expanding the modulus squared yields
\begin{align}
p_\nu^{(s)}(t)-p_\nu(t)
=
2\Big\{
&p_\nu(s)p_\nu(t-s)
\nonumber
\\
&-
\operatorname{Re}\!\left[
A_\nu(t)^*A_\nu(t-s)A_\nu(s)
\right]
\Big\}.
\end{align}
The Kolmogorov inconsistency therefore takes the exact form
\begin{equation}
K_\nu(s,t)
=
2\left|
p_\nu(s)p_\nu(t-s)
-
\operatorname{Re}\!\left[
A_\nu(t)^*A_\nu(t-s)A_\nu(s)
\right]
\right|.
\label{eq:app-exact-K}
\end{equation}

The corresponding two-time correlator is
\begin{align}
C_{12}
&=
P_{++}+P_{--}-P_{+-}-P_{-+}
\nonumber
\\
&=
1
-
2p_\nu(t_1)
\left[
1-p_\nu(t_2-t_1)
\right]
\nonumber
\\
&\quad
-
2\left|
A_\nu(t_2)
-
A_\nu(t_2-t_1)A_\nu(t_1)
\right|^2.
\end{align}

\section{Derivation of the even short-time expansion}
\label{app:short-time}

This appendix provides the derivation of
Eqs.~\eqref{eq:even_K_expansion} and
\eqref{eq:K_quartic_expansion}. 
We write
$s=\alpha t$ and $r=(1-\alpha)t$.

The shift introduced in Eq.~\eqref{eq:centered_definitions} produces only
a global phase:
\begin{equation}
A_\nu(u)
=
e^{-i\mu_1^{(\nu)}u}
\widetilde A_\nu(u),
\qquad
\widetilde A_\nu(u)
=
\bra{\nu}
e^{-i\widetilde H_\nu u}
\ket{\nu}.
\label{eq:app-centered-amplitude}
\end{equation}
Since $s+r=t$, this global phase cancels from both return probabilities
and from the modulus
$\left|
A_\nu(t)-A_\nu(r)A_\nu(s)
\right|$.
We may therefore perform the expansion using
$\widetilde H_\nu$.

The centered return amplitude is
\begin{align}
\widetilde A_\nu(u)
={}&
1
-\frac{\gamma_\nu}{2}u^2
+\frac{i\widetilde\mu_3^{(\nu)}}{6}u^3
+\frac{\widetilde\mu_4^{(\nu)}}{24}u^4
+
O(u^5).
\label{eq:app-centered-amplitude-expansion}
\end{align}
Using
$\widetilde A_\nu(-u)=\widetilde A_\nu(u)^*$, the return probability
contains only even powers:
\begin{equation}
p_\nu(u)
=
\left|\widetilde A_\nu(u)\right|^2
=
1-\gamma_\nu u^2
+\left[
\frac{\widetilde\mu_4^{(\nu)}}{12}
+
\frac{\gamma_\nu^2}{4}
\right]u^4
+
O(u^6).
\label{eq:app-return-probability-expansion}
\end{equation}
The coefficient in square brackets is precisely $\zeta_\nu$ defined in
Eq.~\eqref{eq:chi_definition}.

The interference amplitude appearing in
Eq.~\eqref{eq:exact_measured_return} satisfies
\begin{equation}
\widetilde A_\nu(t)
-
\widetilde A_\nu(r)\widetilde A_\nu(s)
=
-\gamma_\nu sr+
\frac{i\widetilde\mu_3^{(\nu)}}{2}
sr(s+r)
+
O(t^4).
\label{eq:app-interference-amplitude}
\end{equation}
The leading term is real, whereas the cubic term is purely imaginary.
Their cross term therefore vanishes in the modulus squared, giving
\begin{equation}
\left|
\widetilde A_\nu(t)
-
\widetilde A_\nu(r)\widetilde A_\nu(s)
\right|^2
=
\gamma_\nu^2s^2r^2
+
O(t^6).
\label{eq:app-interference-probability}
\end{equation}

Substitution of
Eqs.~\eqref{eq:app-return-probability-expansion} and
\eqref{eq:app-interference-probability} into
Eq.~\eqref{eq:exact_measured_return} yields
\begin{align}
p_\nu^{(s)}(t)-p_\nu(t)
={}&
2\gamma_\nu sr
+
\zeta_\nu
\left(
s^4+r^4-t^4
\right)
+
2\gamma_\nu^2s^2r^2
+
O(t^6).
\label{eq:app-return-difference}
\end{align}
Setting $s=\alpha t$ and $r=(1-\alpha)t$ gives
Eq.~\eqref{eq:K_quartic_expansion}.

To establish the even structure to all orders, define
$\delta_\nu(\alpha t,t)
=
p_\nu^{(\alpha t)}(t)-p_\nu(t)$.
Hermiticity implies
$A_\nu(-u)=A_\nu(u)^*$,
and the exact identity
Eq.~\eqref{eq:exact_measured_return} therefore gives
\begin{equation}
\delta_\nu(-\alpha t,-t)
=
\delta_\nu(\alpha t,t).
\label{eq:app-delta-even}
\end{equation}
Thus $\delta_\nu$ is an even analytic function of $t$. Its leading
coefficient is
$2\gamma_\nu\alpha(1-\alpha)>0$
for $\gamma_\nu>0$ and $0<\alpha<1$. Hence, for sufficiently small
positive $t$, $K_{\nu,\boldsymbol{\Phi}}(\alpha t,t)
=
\delta_\nu(\alpha t,t)$,
which proves Eq.~\eqref{eq:even_K_expansion}.

Finally, expanding the centered fourth moment gives
\begin{equation}
\widetilde\mu_4^{(\nu)}
={}
\mu_4^{(\nu)}
-
4\mu_1^{(\nu)}\mu_3^{(\nu)}
+
6\left(\mu_1^{(\nu)}\right)^2\mu_2^{(\nu)}
-
3\left(\mu_1^{(\nu)}\right)^4.
\end{equation}

\section{Rooted spectral moments and gauge-invariant closed walks}
\label{app:closed-walks}

We provide here the graph-theoretic interpretation of the rooted spectral
moments introduced in Eq.~\eqref{eq:rooted_moments}. Expanding the matrix
product gives
\begin{equation}
\begin{split}
\mu_n^{(\nu)}
&=
\left(H_\chi^n\right)_{\nu\nu}
=\\&\sum_{x_1,\ldots,x_{n-1}}
(H_\chi)_{\nu x_{n-1}}
(H_\chi)_{x_{n-1}x_{n-2}}
\cdots
(H_\chi)_{x_1\nu}.
\label{eq:app-closed-walk-sum}
\end{split}
\end{equation}
Each term is associated with a length-$n$ closed walk
$w=(\nu,x_1,\ldots,x_{n-1},\nu)$ rooted at the measured vertex $\nu$. Its weight is
\begin{equation}
\mathcal{W}_w
=
(H_\chi)_{\nu x_{n-1}}
(H_\chi)_{x_{n-1}x_{n-2}}
\cdots
(H_\chi)_{x_1\nu}.
\label{eq:app-walk-weight}
\end{equation}
Off-diagonal matrix elements describe hopping along graph edges, whereas
diagonal matrix elements may be viewed as stationary steps.

Under the local gauge transformation
$H_\chi
\longmapsto
\Lambda_{\boldsymbol{\beta}}
H_\chi
\Lambda_{\boldsymbol{\beta}}^\dagger$,
the matrix elements transform as
$(H_\chi)_{xy}
\longmapsto
e^{i(\beta_x-\beta_y)}
(H_\chi)_{xy}$.
The gauge factor accumulated by the product
Eq.~\eqref{eq:app-walk-weight} is
\begin{align}
e^{i(\beta_\nu-\beta_{x_{n-1}})}
e^{i(\beta_{x_{n-1}}-\beta_{x_{n-2}})}
\cdots
e^{i(\beta_{x_1}-\beta_\nu)}
=1.
\label{eq:app-gauge-cancellation}
\end{align}
Thus, every individual rooted closed-walk weight is gauge invariant, and
so are the rooted moments obtained by summing these contributions.

For a graph with a chosen basis of independent cycles, the flux-dependent
part of the phase accumulated by a closed walk can be written as
\begin{equation}
\varphi_w^{(\mathrm{flux})}
=
\sum_C n_C(w)\Phi_C
\pmod{2\pi},
\label{eq:app-winding-phase}
\end{equation}
where $n_C(w)\in\mathbb{Z}$ is the oriented winding number of the walk
around the cycle $C$. A segment traversed along an edge and immediately
retraced in the opposite direction accumulates zero net Peierls phase.
Consequently, closed walks consisting only of backtracking and stationary
steps are flux independent, whereas walks with nonzero winding can carry a
nontrivial dependence on the gauge-invariant cycle fluxes.

The centered moments
$\widetilde\mu_n^{(\nu)}
=
\bra{\nu}
\left(
H_\chi-\mu_1^{(\nu)}\mathbb{I}
\right)^n
\ket{\nu}$
are polynomial combinations of the uncentered rooted moments and therefore
inherit their gauge invariance. The occurrence of a flux-dependent walk in
a given moment does not, however, guarantee that the corresponding
observable coefficient is flux dependent. Quantities such as
$K_{\nu,\boldsymbol{\Phi}}$ combine several moments and sums over many
closed walks. Contributions with different orientations or geometries may
cancel because of parity, graph automorphisms, or destructive
interference.

The length of the shortest cycle accessible from $\nu$ therefore provides
only a possible geometric scale for the onset of flux sensitivity. The
actual observable onset is determined by the lowest even order at which
the complete gauge-invariant combination of rooted closed-walk
contributions does not vanish.

Finally, if the graph is a tree, every assignment of edge phases can be
removed by a vertex-dependent gauge transformation. Equivalently, every
closed walk on a tree decomposes into backtracking segments and carries no
physical cycle flux. Hence no transition probability, return statistic,
sequential measurement quantity, or temporal correlator can depend
physically on the Peierls phases.


\section{Winding-number representation and parity-dependent flux onset}
\label{app:cycle-asymptotics}

\paragraph{Winding-number representation.}

Starting from the spectral return amplitude in
Eq.~\eqref{eq:cycle_return_amplitude}, with the eigenvalues
$\lambda_k(\phi)$ given in Eq.~\eqref{eq:cycle_eigenvalues}, we extract
the global phase $e^{-2it}$ and use the Jacobi--Anger identity
\begin{equation}
e^{iz\cos\vartheta}
=
\sum_{m\in\mathbb{Z}}
i^m J_m(z)e^{im\vartheta}
\end{equation}
to obtain
\begin{equation}
A_N(t;\phi)
=
\frac{e^{-2it}}{N}
\sum_{m\in\mathbb{Z}}
i^m J_m(2t)e^{-im\phi/N}
\sum_{k=0}^{N-1}
e^{2\pi i mk/N}.
\end{equation}
The discrete Fourier sum satisfies
\begin{equation}
\frac{1}{N}
\sum_{k=0}^{N-1}
e^{2\pi i mk/N}
=
\begin{cases}
1, & m=\ell N,\quad \ell\in\mathbb{Z},\\
0, & \text{otherwise},
\end{cases}
\end{equation}
and therefore selects the winding sectors $m=\ell N$, yielding
Eq.~\eqref{eq:cycle_winding_expansion}. Combining the sectors $\ell$
and $-\ell$ using
$J_{-n}(z)=(-1)^nJ_n(z)$ and
$i^{-\ell N}(-1)^{\ell N}=i^{\ell N}$ gives
Eq.~\eqref{eq:cycle_winding_cosine}. The flux periodicity and
flux-reversal symmetry in
Eqs.~\eqref{eq:cycle_flux_periodicity}--\eqref{eq:cycle_flux_evenness}
then follow immediately.

\paragraph{Short-time expansion.}

We now derive the leading flux-dependent contribution to the
Kolmogorov inconsistency on the cycle $C_N$. We use
the functions $F_N(\alpha)$ and $G_N(\alpha)$ defined in
Eqs.~\eqref{eq:F_N_alpha} and \eqref{eq:G_N_alpha}, respectively. Both
are positive for $0<\alpha<1$.

For $m\in\mathbb{N}_0$, the small-argument expansion of the Bessel
function of the first kind is
\begin{equation}
J_m(2t)
=
\frac{t^m}{m!}
+
O(t^{m+2}).
\label{eq:app-bessel-small}
\end{equation}
The first flux-dependent winding sectors are $\ell=\pm1$ and enter the
return amplitude at order $t^N$. Whether this contribution survives
linearly in return-based quantities depends on the parity of $N$.

\paragraph{Even cycles.}

For even $N$, the factor $i^N=(-1)^{N/2}$ is real. After removing the
common global phase, the difference between the finite-flux and
zero-flux winding contributions is
\begin{equation}
A_N(t;\phi)-A_N(t;0)
=
2(-1)^{N/2}
\frac{t^N}{N!}
\left(\cos\phi-1\right)
+
O(t^{N+2}).
\label{eq:app-even-amplitude-contrast}
\end{equation}
This correction interferes linearly with the real zero-winding sector.
Using the decomposition in
Eq.~\eqref{eq:app-exact-measured-return}, with $s=\alpha t$ and
$r=(1-\alpha)t$, the amplitude $A_N(t)-A_N(r)A_N(s)$ is $O(t^2)$ at zero
flux, while its flux-dependent correction is $O(t^N)$. Hence its
contribution to the flux contrast starts at $O(t^{N+2})$, and the leading
$O(t^N)$ term comes from the return-probability terms
$p_N(s)p_N(t-s)-p_N(t)$, yielding
\begin{equation}
\Delta_\phi K_N(\alpha t,t)
=
\frac{4(-1)^{N/2}}{N!}
\left(1-\cos\phi\right)
F_N(\alpha)t^N
+
O(t^{N+2}).
\label{eq:app-even-onset}
\end{equation}
Thus, the first flux-dependent contribution to the Kolmogorov
inconsistency is of order $t^N$ for even cycles.

\paragraph{Odd cycles.}

For odd $N$, the factor $i^N$ is purely imaginary. The single-winding
contribution therefore does not interfere linearly with the real
zero-winding amplitude in the return probabilities. The order-$t^N$
flux correction cancels, and the first nonzero contribution appears at
order $t^{2N}$.

Up to this order, and again omitting the common global phase, the return
amplitude can be written as
\begin{align}
A_N(t;\phi)
={}&
J_0(2t)
+
2i\eta_N J_N(2t)\cos\phi
\nonumber
\\
&-
\frac{2t^{2N}}{(2N)!}\cos(2\phi)
+
O(t^{2N+2}),
\label{eq:app-odd-A}
\end{align}
where
$\eta_N=(-1)^{(N-1)/2}$.
Keeping the zero-winding and single-winding Bessel terms unexpanded at
this stage prevents the loss of subleading contributions below order
$t^{2N+2}$.

Using Eq.~\eqref{eq:app-bessel-small} only after forming the return
probabilities and the interference term, substitution into
Eq.~\eqref{eq:app-exact-K} gives
\begin{equation}
\begin{split}
\Delta_\phi K_N(\alpha t,t)
=
\sin^2\phi
\Bigg[
&
\left(
\frac{4}{(N!)^2}
-
\frac{8}{(2N)!}
\right)
G_N(\alpha)
\\
&-
\frac{4}{(N!)^2}
F_N(\alpha)^2
\Bigg]
t^{2N}
+
O(t^{2N+2}).
\end{split}
\label{eq:app-odd-onset}
\end{equation}
Hence, for odd cycles, the flux response is delayed from order $t^N$ in
the return amplitude to order $t^{2N}$ in the return-based Kolmogorov
inconsistency.

\section{Numerical methods and figure generation} \label{app:numerics} 
All numerical results were obtained from the exact finite-size spectral
return amplitude in Eq.~\eqref{eq:cycle_return_amplitude}, with
\(p_N(t;\phi)=|A_N(t;\phi)|^2\).  
All plotted quantities were evaluated from the full finite-size spectral
sum, without using truncated short-time or winding-number expansions.
\paragraph{Finite-time measurement back-action.} \label{app:numerics_fig2} 
Figure~\ref{fig:finite-time-backaction} was generated from the midpoint
contrast defined in Eq.~\eqref{eq:cycle_midpoint_contrast}.
For \(N=6\), we sampled \(0\leq t\leq10\) and
\(0\leq\phi\leq\pi\) on a uniform \(700\times500\)
time--flux grid. The color scale was centered at zero using a symmetric normalization over $ \left[ -\max|\Delta_\phi K_6^{\mathrm{mid}}|, \max|\Delta_\phi K_6^{\mathrm{mid}}| \right]. $ Positive and negative values therefore represent, respectively, flux-induced enhancement and suppression relative to the zero-flux dynamics. The black curves are the numerically determined contours $\Delta_\phi K_6^{\mathrm{mid}}(t)=0$.
\paragraph{Leggett--Garg phase diagrams and optimization.} \label{app:numerics_figs34} The results in Figs.~\ref{fig:L3-phase-diagram-C10} and \ref{fig:L3-strength-and-timing} were obtained by substituting the exact return amplitude in Eq.~\eqref{eq:cycle_return_amplitude} into the Leggett--Garg functional in Eq.~\eqref{eq:L3_equally_spaced}. For Fig.~\ref{fig:L3-phase-diagram-C10}, we evaluated
\(L_{3,10}(\tau;\phi)\) on a uniform \(501\times4001\) grid
over \(0\leq\phi\leq\pi\) and \(0\leq\tau\leq40\). The color map displays \([L_{3,10}-1]_+\), the black curves are the numerical contours \(L_{3,10}=1\), and the open marker denotes the largest value found on the displayed grid. For Fig.~\ref{fig:L3-strength-and-timing}(a), the flux interval was sampled at \(501\) uniformly spaced points. At each flux, the temporal maximization was initialized from \(20001\) uniformly spaced values over \(0\leq\tau\leq T\), with \(T=24\) for \(C_6\) and \(T=40\) for \(C_{10}\). Local maxima were first identified on the discrete curve. The twelve highest candidate peaks were then refined by bounded one-dimensional maximization over their corresponding local intervals. The largest refined value was retained as \(L_{3,N}^{\star}(\phi;T)\), with an absolute optimization tolerance of \(10^{-11}\). For Fig.~\ref{fig:L3-strength-and-timing}(b), the curves at \(\phi=0\) and \(\phi=\pi\) for \(C_{10}\) were sampled at \(40001\) uniformly spaced times over $ 0\leq\tau\leq40$. The target value was  $L_{\mathrm{target}} = 1+0.99 \left[ L_{3,10}^{\star}(0;40)-1 \right]$. The earliest threshold crossing was first bracketed on the discrete time grid and then refined using Brent's root-finding method, with absolute and relative tolerances \(10^{-11}\) and \(10^{-12}\), respectively. The refined crossing times were used to evaluate the speedup ratio in Eq.~\eqref{eq:speedup}.

\bibliography{bibliography}

\end{document}